\documentclass[smallextended]{llncs}
\usepackage[utf8]{inputenc}
\usepackage{ amssymb }
\usepackage{amsmath}
\usepackage{enumerate}
\usepackage{enumitem}
 \usepackage[linesnumbered,boxed]{algorithm2e}

\usepackage[margin=1in]{geometry}

\usepackage{ stmaryrd }

\usepackage{mathtools}

 \newtheorem{nclaim}{Claim}

\newcommand{\eu}[3]{e_{#1}({#2}, {#3})}
\newcommand{\eus}[1]{e_{#1}}
\newcommand{\eup}[1]{e'_{#1}}

\newtheorem{observation}{Observation}

\newcommand{\B}{\mathcal{B}}

\begin{document}
\title{Improved Kernelization and Fixed-parameter Algorithms for Bicluster Editing}
\author{Manuel Lafond}
\authorrunning{M. Lafond}
\institute{Université de Sherbrooke, Canada, 
\email{manuel.lafond@USherbrooke.ca}}
\maketitle              
\begin{abstract}
Given a bipartite graph $G$, the \textsc{Bicluster Editing} problem asks for the minimum number of edges to insert or delete in $G$ so that every connected component is a bicluster, i.e. a complete bipartite graph.  This has several applications, including in bioinformatics and social network analysis.  In this work, we study the parameterized complexity under the natural parameter $k$, which is the number of allowed modified edges.
We first show that one can obtain a kernel with $4.5k$ vertices, an improvement over the previously known quadratic kernel.
We then propose an algorithm that runs in time $O^*(2.581^k)$.  Our algorithm has the advantage of being conceptually simple and should be easy to implement.

\vspace{1mm}

\vspace{1mm}
\noindent
\textbf{Special acknowledgments:}
The author acknowledges the contributions of Benoît Charbonneau and Pierre-Luc Parent for their help in finding the $4.5k$ lower bound on the kernel size.
\end{abstract}

%\newpage

\section{Introduction}

Partitioning data points into clusters, which are often interpreted as groups of similarity, is a fundamental task in computer science with practical applications in many areas.
Although several formulations of what constitutes a good clustering have been proposed (e.g. based on pairwise distances~\cite{liu2012editing,hartung2015structural,dondi2023tractability}, modularity~\cite{newman2004finding}, random walks~\cite{rosvall2008maps} or likelihood~\cite{karrer2011stochastic}), most approaches are based on the principle that a cluster should contain members that are similar to each other, and different from the members outside of the cluster.  In a graph-theoretic setting, the ideal clustering should therefore consist of disjoint cliques.  

In some applications, one has two classes of data points and only the relationships \emph{between} classes are of interest.  Clustering then takes the form of finding sub-groups in which the members of one class have similar relationships (e.g. groups of people who like the same movies).  This can be modeled as a bipartite graph with two vertex sets $V_1$ and $V_2$, one for each class, and the goal is to  partition $V_1 \cup V_2$ into disjoint \emph{biclusters}, which are complete bipartite graphs.  
This variant has several applications~\cite{tanay2005biclustering,madeira2004biclustering}, for instance in the analysis of social interactions between groups~\cite{barber2007modularity,gnatyshak2012gaining}, gene expression data~\cite{cheng2000biclustering,pontes2015biclustering}, and phylogenetics (e.g. when comparing the left and right descendants of an ancestral species \cite{altenhoff2019oma,lafond2018accurate}).

In this work, we assume that bipartite graphs that do not consist of 
disjoint biclusters are due to erroneous edges and non-edges.  In the \textsc{Bicluster Editing} problem, we ask whether a given bipartite graph $G$ can be transformed into a set of disjoint biclusters by adding/removing at most $k$ edges (we forbid inserting edges between vertices on the same side of the bipartition).  
%We focus on the parameterized complexity of the problem with respect to parameter $k$.  

\vspace{2mm}

\noindent
\textbf{Related work.}  If the input graph is not required to be bipartite, \textsc{Bicluster Editing} is known to be NP-hard even on graphs of maximum degree $3$~\cite{drange_et_al:LIPIcs.IPEC.2015.402} and on dense graphs~\cite{sun2014complexity}.  The former result holds even if only edge deletions are allowed, and also implies that under the Exponential Time Hypothesis, no subexponential time algorithm is possible, either with respect to the number of edges to edit or the number of vertices and edges of the graph.  
If the input graph is required to be bipartite and edges cannot be added between vertices in the same side of the given bipartition, the problem is still NP-hard~\cite{amit2004bicluster}.
To our knowledge, for this variant it is unknown whether hardness also holds on bipartite graphs of bounded maximum degree.

Nonetheless, on bipartite graphs, by observing that graphs of disjoint biclusters coincide with bipartite $P_4$-free graphs, Protti et al.~\cite{protti2006applying} first devised a simple $O^*(4^k)$ time algorithm that finds a $P_4$, and branches over the four possible ways to remove it (here, the $O^*$ notation suppresses polynomial factors).  They also show that \textsc{Bicluster Editing} admits a kernel of $4k^2 + 6k$ vertices.  
In~\cite{guo2008improved}, the authors extend the branching algorithm to solve induced $P_5$'s, if any --- this leads to slightly more case handling but achieves an improved running time of $O^*(3.24^k)$.  A kernel of size $4k$ is also proposed, but this turns out to be a slight inaccuracy (we provide a counter-example).  Note that on graphs that are not necessarily bipartite, an algorithm in time $O^*(2^{5\sqrt{pk}})$ is proposed in~\cite{drange_et_al:LIPIcs.IPEC.2015.402}, where $p$ is the number of desired biclusters.  Going back to the bipartite case, in terms of approximability, a factor $11$ approximation algorithm is presented in~\cite{amit2004bicluster} and is improved to a pivot-based, randomized factor $4$ approximation in~\cite{ailon2012improved}.
Owing to the practical applications of the problem, several heuristics and experiments have also been published~\cite{de2012hybrid,pinheiro2016solving,de2017new,sun2013biclue}.
Recently, Tsur found a $O^*(2.636^k)$ time algorithm for the problem~\cite{tsur2021faster}, and then a $O^*(2.22^k)$ time algorithm~\cite{tsur2023faster}.  Currently, this is the fastest known algorithm for the problem, and was found using automated techniques to evaluate the branching factors for a relatively large number of small graphs.

A closely related problem is \textsc{Cluster Editing}, where the given graph does not have to be bipartite and the goal is to attain a collection of disjoint cliques.  \textsc{Cluster Editing} is known to be NP-hard on graphs of maximum degree at least $6$~\cite{komusiewicz2012cluster}, and to our knowledge the complexity is still open for maximum degree between $3$ and $5$.  The problem admits an $O^*(1.619^k)$ FPT algorithm~\cite{bocker2012golden} (obtained after a series of improvements~\cite{gramm2003graph,gramm2004automated}), and a 3-approximation algorithm~\cite{bansal2004correlation} (a 2.5-approximation is given in~\cite{van2009deterministic} if edit weights satisfy probability constraints).  
In terms of kernelization, the problem admits a kernel with $2k$ vertices~\cite{chen20122k}.  This bound also holds when each possible edge operation is assigned an integral weight~\cite{cao2012cluster}.

%Unpublished manuscripts
%Faster parameterized algorithm for Bicluter Editing
%$O^*(3.116^k)$.

%$O(k)$ problem kernel, although constants not discussed
%https://arxiv.org/pdf/1506.00944.pdf

\vspace{2mm}

\noindent
\textbf{Our contributions.}
We provide a kernel with at most $4.5k$ vertices for \textsc{Bicluster Editing} on bipartite graphs using three relatively simple reduction rules.  
This is an improvement over the quadratic kernel of~\cite{protti2006applying}.
The proof on the kernel size is based on a combinatorial charging argument.
We also show that our analysis is tight, in the sense that our reduction rules cannot achieve a smaller kernel.
  Our first two rules are the same as those proposed in~\cite{guo2008improved}, and we show that the latter cannot achieve a kernel better than $6k$.  
  
We then provide a fixed-parameter branching algorithm that runs in time $O^*(2.581^k)$.  
The algorithm is the same as the one from the conference version of this paper~\cite{lafond2020even}, where 
complexity $O^*(2.695^k)$ was achieved. The improvement in complexity results from a more refined analysis of the algorithm.
The $O^*(2.22^k)$ time algorithm of Tsur~\cite{tsur2023faster} appeared after the publication of our conference version, and some of the techniques from this version were used to achieve this result (notably, our maximum-deletion solution from Lemma~\ref{lem:solution-with-most-dels} and our degree 1 rule from Lemma~\ref{lem:degree-one} are used in Tsur's algorithm).  
In comparison, the main advantage of our algorithm is conceptually simple, since it requires no automated search nor complicated branching rules. 
Indeed, we use a simple idea: two vertices are in conflict if their neighborhood intersects but have different neighbors.  We find two conflicting vertices, and branch into every way of eliminating the conflict, which requires making their neighborhood disjoint or identical.
This is a simple but powerful idea that yields a $O^*(3.237^k)$ time algorithm with a very simple analysis.  We then show that by incorporating a branching rule to handle degree $1$ vertices and by tightening our analysis, we can achieve time $O^*(2.581^k)$.  We mention that the $O^*(2.22^k)$ algorithm of Tsur can be seen as a generalization of this idea: instead of resolving conflicts between two vertices, one can take a larger set $X$ of vertices, and branch into every way of eliminating every conflict vertices of $X$.  This requires branching over every way of partitioning $X$ and analyzing the worst-case branching factor, which is (quite roughly speaking) one of the ideas proposed in~\cite{tsur2023faster}.

\section{Preliminary notions.}
For an integer $n$, we use the notation $[n] = \{1,2,\ldots,n\}$ (which is the empty set if $n \leq 0$).  Given two sets $A$ and $B$, we write 
$A \triangle B = (A \setminus B) \cup (B \setminus A)$ for the \emph{symmetric difference} between $A$ and $B$.
Let $G = (V, E)$ be a graph and let $v \in V$. 
We write $N_G(v)$ for the set of neighbors of $v$ and define $deg_G(v) = |N_G(v)|$.  The subscript $G$ may be dropped if it is clear from the context.
For $F$ a set of unordered pairs of $V$, we write $G - F$ for the graph $(V, E \setminus F)$.  If $F = \{xy\}$ has a single element, we may write $G - xy$ instead of $G - \{xy\}$.
%We write ${ V \choose 2 }$ for the set of unordered pairs of vertices of $G$.  For $F \subseteq { V \choose 2}$,
%we define $G \triangle F = (V, E \triangle F)$, i.e. the graph obtained by deleting the edges in $E \cap F$, and inserting the edges in $F \setminus E$.

The vertex set $V(G)$ of a bipartite graph $G = (V_1 \cup V_2, E)$ has two disjoint subsets $V_1$ and $V_2$ which are independent sets.
For $X \subseteq V_i$, $i \in \{1,2\}$, we define $N_G(X) = \bigcup_{x \in X} N_G(x)$.
%A \emph{bicluster} is a bipartite graph in which every possible edge between $V_1$ and $V_2$ is present.
Two vertices $u, v \in V_i$, with $i \in \{1,2\}$ are called \emph{twins} if $N(u) = N(v)$.
%\footnote{In some definitions of twins, one requires $uv$ to be an edge, but not here.}.
Note that twins form an equivalence relation.  A \emph{twin class} is an equivalence class, i.e. 
a maximal subset $X \subseteq V(G)$ such that $x_1$ and $x_2$ are twins for every $x_1, x_2 \in X$ (note that a twin class may consist of a single vertex).
%Given two graphs $G_1$ and $G_2$, we will use the shorthand $G_1 \triangle G_2 := E(G_1) \triangle E(G_2)$.

Given a bipartite graph $G = (V_1 \cup V_2, E)$ and subsets $X \subseteq V_1$ and $Y \subseteq V_2$, we say that $X \cup Y$ form 
a \emph{bicluster} if $N(x) = Y$ for all $x \in X$ and $N(y) = X$ for all $y \in Y$ (note that an isolated vertex is a bicluster).   A graph is a \emph{bicluster graph} if each of its connected components is a bicluster.
We will often use $\B$ to denote a bicluster graph.
For bipartite graph $G = (V_1 \cup V_2, E)$, we say that a bicluster graph $\B$ is a \emph{biclustering of $G$} if $V(\B) = V(G)$ and all edges of $\B$ have an endpoint in each of $V_1$ and $V_2$.  The \emph{cost} of $\B$ (with respect to $G$) is $|E(G) \triangle E(\B)|$.  The edges in $E(G) \setminus E(\B)$ are \emph{deletions}, and the edges in $E(\B) \setminus E(G)$ are \emph{insertions}.

The \textsc{Bicluster Editing} problem asks, for a given a bipartite graph $G$ and a non-negative integer $k$, whether there exists a biclustering $\B$ of $G$ 
such that ${|E(G) \triangle E(\B)| \leq k}$.  We say that $\B$ is \emph{optimal} if no biclustering of $G$ has cost strictly less than $\B$.

For $u \in V(G)$, we denote by $\eu{u}{G}{\B}$ the number of elements of $E(G) \triangle E(\B)$ that contain $u$ (i.e. $\eu{u}{G}{\B}$ counts the number of inserted and deleted edges incident to $u$).  

\begin{observation}\label{obs:sumcost}
Let $G = (V_1 \cup V_2, E)$ be a bipartite graph and let $\B$ be a biclustering of $G$.  Let $i \in \{1,2\}$.  
Then $|E(G) \triangle E(\B)| = \sum_{u \in V_i} \eu{u}{G}{\B}$.
\end{observation}

\begin{proof}
    Each edge in $E(G) \triangle E(\B)$ has exactly one endpoint in $V_i$.  Thus by summing the $\eu{u}{G}{\B}$ values for each $u$ of $V_i$, we count each modified edge exactly once.
    \qed
\end{proof}

\section{A $4.5k$ kernel}

%Our main algorithm makes use of the reduction rules for our $5k$ kernel.  We therefore present our kernelization first.
From here on, we assume that $G = (V_1 \cup V_2, E)$ is a given bipartite graph and that $k$ is a positive integer (if $k = 0$, we can easily verify that $G$ is itself a bicluster graph).
We provide three reduction rules based on those proposed in~\cite{guo2008improved}.
We use the notion of a \emph{sister} of a vertex $u$.  %, where a sister is a twin-less vertex that has the same neighbors as $u$, plus one more.  
%More precisely, 
Let $R$ be a twin class of $G$.
Let $S = N(R)$ and $t \in N(S) \setminus R$.
Then $t$ is a \emph{sister} of $R$ if the two following conditions are satisfied:

\begin{itemize}
    \item
    $t$ does not have a twin in $G$ (i.e. $\{t\}$ is a twin class);
    
    \item 
    $N(t) = S \cup \{v\}$ for some vertex $v \in V(G) \setminus S$.
\end{itemize}

In words, $t$ is a sister of $R$ if $t$ has the same neighbors as every element of $R$, with the exception of an extra vertex $v$, and $t$ has no twin.  See Figure~\ref{fig:twins}.

\begin{figure}
    \centering
    \includegraphics[width=0.4\textwidth]{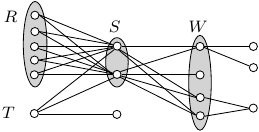}
    \caption{$R$ is a twin class, $S = N(R)$, $T$ are the sisters of $R$ (there is only one) and $W = N(S) \setminus (T \cup R)$.  Note how the two bottom vertices of $W$ are not sisters of $R$, because they are twins.}
    \label{fig:twins}
\end{figure}

For $u \in V(G)$, we say that $t$ is a sister of $u$ if $t$ is a sister of the twin class containing $u$. 
We can now describe our reduction rules.

\vspace{2mm}

\noindent 
\textbf{Rule 1}: if a connected component $X$ of $G$ is a bicluster, remove $X$ from $G$.

\vspace{2mm}

\noindent 
\textbf{Rule 2}: if there is a twin class $R$ such that $|R| > |N(N(R)) \setminus R|$, then remove any vertex from $R$.

\vspace{2mm}

\noindent 
\textbf{Rule 3}: let $R$ be a twin class, let $T$ be the set of sisters of $R$, and let $W = N(N(R)) \setminus (R \cup T)$. 
If $|R| > |W|$ and $|T| \geq 1$, then choose any $t \in T$ and remove the edge between $t$ and $N(t) \setminus N(R)$.   Then reduce $k$ by $1$.

\vspace{2mm}

The idea of Rule 3 is that if removing a single edge from a vertex of $T$ would help that vertex join a ``large enough'' twin-class, then we should do so.
For instance in Figure~\ref{fig:twins}, according to Rule 3, we can remove the edge between the vertex of $T$ and the bottom degree $1$ vertex, making the size of the  twin class $R$ larger. % (making Rule 2 subsequently applicable).

Note that in~\cite{guo2008improved}, Rule 1 and Rule 2 were introduced and the authors suggested that they would lead to a kernel of size $4k$.  However, the $P_6$ graph is a counter-example to this claim (the $P_6$ is the chordless path with $6$ vertices).  Indeed, the $P_6$ admits a bicluster graph with $k = 1$, has $6k$ vertices, and Rule 1 and Rule 2 do not apply.  
It is plausible that Rule 1 and Rule 2 can lead to a $6k$ kernel, whereas Rule 3 allows us to reduce this to $4.5k$ (in particular, Rule 3 does reduce the $P_6$ when we consider $R$ containing a degree $1$ vertex).  
Observe that for any $k \geq 1$, taking $k$ disjoint unions of $P_6$ graphs results in a graph with $6k$ vertices that requires $k$ edge modifications.
%For instance on the $P_6$, each vertex of degree one has a sister and $|W| = 0$, letting us remove an edge and then the whole graph using Rule 1.

It is easy to see that Rule 1 is safe.  
Rule 2 was already shown to be safe in~\cite{guo2008improved}, but we include a proof here for completeness.  The proof was based on a result on twins that we restate here, as it will be useful later.

\begin{lemma}[\cite{guo2008improved}]\label{lem:twins-in-same}
There exists an optimal biclustering $\B$ of $G$ such that, for every twin class $R$ of $G$, the vertices of $R$ are all in the same bicluster of $\B$.
\end{lemma}

\begin{proof}
Let $\B$ be an optimal biclustering of $G$.  For a vertex $w \in V(G)$, we write  $\eus{w} := \eu{w}{G}{\B}$ for the number of modified edges incident to $w$.
Suppose that there is a twin class $R$ such that not all elements of $R$ are in the same bicluster of $\B$. 
Let $u$ be a vertex of $R$ that minimizes $\eus{u}$, and let $B_u$ be the set of vertices of the bicluster that contains $u$.  Let $v \neq u$ be a vertex of $R$ such that $v$ is not in the same bicluster as $u$ in $\B$.  
Consider the biclustering $\B'$ obtained from $\B$ by removing $v$ from its bicluster, and adding it to $B_u$ instead.  For $w \in V(G)$, denote $\eup{w} := \eu{w}{G}{\B'}$.
By Observation~\ref{obs:sumcost}, we infer that the change in cost $|E(G) \triangle E(\B')| - |E(G) \triangle E(\B)|$ is equal to $\eup{v} - \eus{v}$, since all vertices on the same side as $v$ have the same neighbors in either $\B$ or 
$\B'$, except $v$.  Since $u$ and $v$ are twins, we have $\eup{v} = \eup{u} = \eus{u}$, and, since $\eus{u} \leq \eus{v}$ by the choice of $u$, we have $\eup{v} - \eus{v} \leq \eus{u} - \eus{u} = 0$.  
In other words, after moving $v$, the cost has not increased, so $\B'$ is optimal as well.  We may repeat this procedure for every $v$ of $R$ not in the same bicluster as $u$.  Once this is done, since every time we apply this procedure, the number of twins that are in the same bicluster increases, we may repeat this for every twin class $R$, yielding the result.
\qed
\end{proof}

We can now argue that our reduction rules preserve equivalent instances.  Recall that a reduction rule is \emph{safe} when the original instance is a \emph{yes} instance if and only if the modified instance is a \emph{yes} instance.

\begin{lemma}
    Rule 2 is safe.
\end{lemma}

\begin{proof}
    To prove the statement, suppose that $G$ has a twin class $R$ such that ${|R| > |N(N(R)) \setminus R|}$, and let $u \in R$.  
    We show that there exists a biclustering of $G$ of cost at most $k$ if and only if there exists a biclustering of $G - u$ of cost at most $k$.

    The forward direction is easy to see as removing vertices cannot increase the number of required modifications, so we focus on the other direction.
    If $|R| = 1$, then $|R| > |N(N(R)) \setminus R|$ implies that the latter quantity is $0$, so that members of $N(R)$ only have the vertex of $R$ in their neighborhood.  In this case, $R \cup N(R)$ induces a bicluster and we may remove it by Rule 1.  In particular, we may remove $u$ without affecting the cost.  Hence, we will assume that $|R| \geq 2$.

    Let $u \in R$.
    Denote $G' := G - u$ and $R' := R \setminus \{u\}$.  Note that $R'$ is a non-empty twin class of $G'$.
    Let $\B$ be a biclustering of $G'$ of cost at most $k$ such that every element of $R'$ is in the same bicluster, which exists by Lemma~\ref{lem:twins-in-same}.
    Suppose without loss of generality that $R' \subseteq V_1$.
    We argue that we may assume that no 
    modified edge is incident to a vertex of $R'$.  Let $B$ be the set of vertices of the bicluster of $\B$ containing $R'$.  Then let $S := N_G(R)$ (which is equal to $N_{G'}(R')$), $S' := B \cap S$, $W := N_G(S) \setminus R$, $W' := B \cap W$, and $Z := B \cap (V_2 \setminus S')$.
    Note for later reference that $|R| > |N(N(R)) \setminus R| = |W|$ implies $|R'| \geq |W|$.

    Consider the biclustering $\B'$ obtained from $\B$ as follows:
    \begin{itemize}
        \item 
        remove the vertices of $R' \cup S \cup W'$ from their respective biclusters;

        \item 
        add $R' \cup S \cup W'$ as a new bicluster.
    \end{itemize}
    Among the modifications in $E(G) \triangle E(\B)$ that are not in $E(G) \triangle E(\B')$, we count at least $|R'||Z| + (|S| - |S'|)|R'|$ modifications, respectively for insertions between $R'$ and $Z$, and deletions between $R'$ and $S \setminus S'$ (see Figure~\ref{fig:rule2proof}). 

    \begin{figure}
        \centering
        \includegraphics[width=\textwidth]{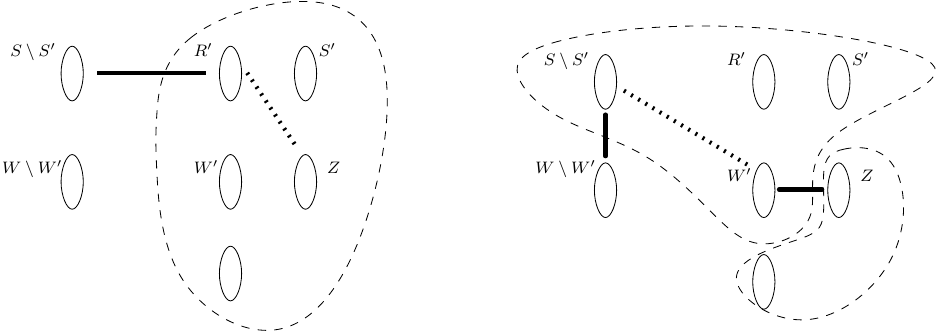}
        \caption{Left: an illustration of $\B$ and relevant subsets of vertices.  The dashed shape indicates $B$, the bicluster that contains $R'$.  The fat edge represents deletions in $\B$ not in $\B'$, and the fat dotted edge insertions in $\B$ not in $\B'$.  Right: an illustration of $\B'$.  The upper dashed shape indicates the new bicluster to introduce, and the lower dashed shape is what remains in $B$ after modification.  Fat edges are the worst-case deletions in $\B'$ not in $\B$, the dotted edge represents insertions in $\B'$ not in $\B$.}
        \label{fig:rule2proof}
    \end{figure}
    
    Consider now the number of modifications in $E(G) \triangle E(\B')$ but not in $E(G) \triangle E(\B)$.  By Observation~\ref{obs:sumcost}, it suffices to look at vertices of $V_2$ whose neighborhood has changed from $\B$ to $\B'$. 
    Note that only vertices of $S$ and $Z$ could have changed neighborhood (if we take a vertex of $V_2 \setminus (S \cup Z)$, none of its neighbors in $\B$ is in $R' \cup W'$ and thus it has not changed neighbors in $\B'$).  The vertices of $S'$ do not incur additional cost since they were only separated from some of their non-neighbors that were in $B$ (bottom subset in Figure~\ref{fig:rule2proof}).
    Each vertex of $S \setminus S'$ is incident to at most $|W| - |W'|$ deletions between $S \setminus S'$ and $W \setminus W'$ that are in $\B'$ but not in $\B$, and at most $|W'|$ insertions between $S \setminus S'$ and $W'$.  Hence in total the vertices of $S$ are incident to at most $|W|(|S| - |S'|)$ modified edges in $\B'$ not in $\B$.
    Each vertex of $Z$ is incident to at most $|W'|$ deletions in $\B'$ not in $\B$, and aside from $R'$, the rest of the neighborhood of $Z$ is unchanged.

    It follows that the number of modifications in $\B'$ not in $\B$ is at most 
    \begin{align*}
    &~|W|(|S| - |S'|) + |W'||Z|
    \leq |R'|(|S| - |S'|) + |R'||Z|
    \end{align*}
    since $|R'| \geq |W| \geq |W'|$.
    This means that from $\B$ to $\B'$, we save at least as many modifications as we create, and hence $\B'$ is also optimal.  Importantly, no vertex of $R'$ is incident to a modified edge in $\B'$.

    So, let us consider the biclustering $\B'$.  We add $u$ to the bicluster that contains $R'$.  Since vertices of $R'$ are unmodified, introducing $u$ has no cost, and thus we obtain a biclustering of $G$ that also has cost at most $k$.
    \qed
\end{proof}

%Only the safety of Rule 3 needs to be proved.

We now proceed to the safety of Rule 3 with a similar but more involved argument.

\begin{lemma}\label{lem:rule3-is-fine}
Rule 3 is safe.
\end{lemma}

\begin{proof}
Let $R$ be a twin class of $G$ and let $T$ be the set of sisters of $R$.  
Also let $S := N(R)$ and $W := N(S) \setminus (R \cup T)$.  Suppose that  $|R| > |W|$ and $|T| \geq 1$. 
Let $u \in T$ and let $v$ be the vertex of $N(u) \setminus S$.
We show that there exists a biclustering of $G$ of cost at most $k$ if and only if there exists a biclustering of $G - uv$ of cost at most $k - 1$.

Suppose that $G - uv$ admits a biclustering of cost at most $k - 1$.  Then we can transform $G$ into $G - uv$ with one deletion, and use $k - 1$ modification to obtain a bicluster graph, which shows that $G$ admits a biclustering with at most $k$ modifications. 

Conversely, suppose that $G$ admits a biclustering of cost at most $k$.  We want to show that $G - uv$ admits a biclustering of cost at most $k - 1$.   We prove a more general statement: that there exists an optimal biclustering of $G$ in which all edges between vertices of $T$ and their extra neighbor not in $S$ are deleted.
Let $\B$ be a optimal biclustering of $G$, whose cost is at most $k$, in which all vertices of $R$ are in the same bicluster (this exists by Lemma~\ref{lem:twins-in-same}).
Let $B_1, \ldots, B_l$ be the sets of vertices in the biclusters of $\B$.
Assume without loss of generality that $R \subseteq V_1$, and that $R \subseteq B_1$.  
For $i \in [l]$, define $T_i := B_i \cap T, W_i := B_i \cap W, Y_i := (B_i \cap V_1) \setminus (R \cup T_i \cup W_i)$.  Also define $S_i := B_i \cap S$, $U_i := B_i \cap (N(T_i) \setminus S)$, and $Z_i := (B_i \cap V_2) \setminus (S_i \cup U_i)$.  Note that $U_i$ only contain the extra neighbors of the $T_i$ vertices in $B_i$.  Therefore, every vertex of $U_i$ has a neighbor in $T_i$, and $Z_i$ may intersect with $N(T) \setminus U_i$.  See Figure~\ref{fig:rule3proof} for an illustration of all the sets for $i=1$ and $i = 2$.

Consider the biclustering $\B'$ obtained
from $\B$ as follows:
\begin{itemize}
\item 
        remove the vertices of $R \cup T \cup S \cup W_1$ from their respective biclusters;

        \item 
        add $R \cup T \cup S \cup W_1$ as a new bicluster.
\end{itemize}

\begin{figure}
    \centering
    \includegraphics[width=0.5\textwidth]{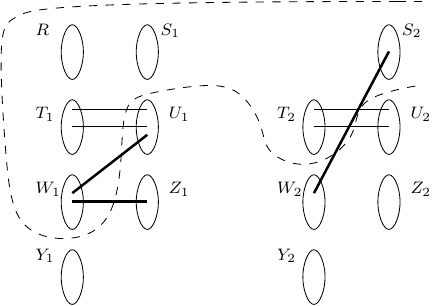}
    \caption{The bicluster introduced in $\B'$ (dashed shape).  Fat edges are deletions in $\B'$ not in $\B$, light edges represents the fact that each element of each $T_i$ has at most one neighbor in $U_i$, and they are all distinct.  Not shown: insertions between $W_1$ and $S \setminus S_1$.}
    \label{fig:rule3proof}
\end{figure}

We argue that $\B'$ is also optimal.  
Our strategy is as follows: each edge modification in $E(G) \triangle E(\B')$ will be \emph{mapped} to a distinct edge modification in $E(G) \triangle E(\B)$ (i.e., we provide an injection from $E(G) \triangle E(\B')$ to $E(G) \triangle E(\B)$).  This will show that $\B'$ does not require more modifications than $\B$.

For an insertion/deletion $xy$ that is present in both $E(G) \triangle E(\B)$ and $E(G) \triangle E(\B')$, we map $xy$ to $xy$.  It thus suffices to consider modifications that are in $E(G) \triangle E(\B')$ but not in $E(G) \triangle E(\B)$.
By Observation~\ref{obs:sumcost}, we may consider the edge modifications incident to vertices of $V_1$ and,  
moreover, it suffices to consider the vertices of $V_1$ whose neighborhood has changed from $\B$ to $\B'$, as the other vertices are only incident to modifications in both sets. 
To map those, let us first list some of the edge modifications that are in $E(G) \triangle E(\B)$ but not in $E(G) \triangle E(\B')$, which we divide into four groups:
\begin{enumerate}[label=(A\arabic*)]

    \item \label{grp:ruz}
    there are $|R|(|U_1| + |Z_1|)$ insertions between vertices of $R$ and $U_1 \cup Z_1$.  Recalling that $|R| > |W|$, this amounts to at least $(|W| + 1)(|U_1| + |Z_1|) \geq |W|(|U_1| + |Z_1|) + |U_1|$ insertions.  

    \item \label{grp:rs}
    there are $|R|(|S| - |S_1|)$ deletions between vertices of $R$ and $S$.  
    This amounts to at least $(|W| + 1)(|S| - |S_1|) \geq |W|(|S| - |S_1|) + |S| - |S_1|$ deletions.

    \item \label{grp:ts}
    for $i \in [l]$, there are $|T_i|(|S| - |S_i|)$ deletions between vertices of $T_i$ and $S$ (because for $t \in T_i$, $S \subseteq N(t)$ by the definition of a sister);

    \item \label{grp:tu}
    for $i \in [l]$, if $|U_i| \geq 1$ there are at least $|T_i|(|U_i| - 1)$ insertions between vertices of $T_i$ and $U_i$ (this is because by the definition of sister vertices, each $T_i$ has at most one neighbor in $U_i$, and thus each vertex of $T_i$ requires at least $|U_i| - 1$ insertions).

\end{enumerate}
Note that the list may include other modifications, but these will suffice for our purposes.
Now let us list \emph{all} the modifications in $E(G) \triangle E(\B')$ that are not in $E(G) \triangle E(\B)$, again focusing on the $V_1$ side, and describe our map for each of them.  
\begin{itemize}
    \item
    \emph{(modifications incident to $R$)}: 
    no such modified edge is incident to a vertex of $R$, because they require no modification at all in $\B'$;

    \item 
    \emph{(modifications incident to $Y_i$'s)}: 
    for $i \in [l]$, no such modified edge is incident to a vertex of $Y_i$, because such a vertex has no neighbor in $S_i$, and only $S_i$ was removed from $B_i$ to go from $\B$ to $\B'$;

    \item 
    \emph{(modifications incident to $W_1$)}: 
    $\B'$ requires at most $|W_1|(|U_1| + |Z_1|)$ deletions between vertices of $W_1$ and $U_1 \cup Z_1$ that were not in $E(G) \triangle E(\B)$.  We map those to distinct modifications of $E(G) \triangle E(\B)$ from group~\ref{grp:ruz}.  

    Notice that there remains at least $|U_1|$ modifications in group~\ref{grp:ruz} that we can still map to.

    There are also at most $|W_1|(|S| - |S_1|)$ insertions between $W_1$ and $S \setminus S_1$.  We map those in the next step.

    \item 
    \emph{(modifications incident to other $W_i$'s)}: 
    for $i \in \{2, \ldots, l\}$,
    there are at most $|W_i||S_i|$ deletions between vertices of $W_i$ and $S_i$.
    Summing over every $i$, this amounts to at most $\sum_{i=2}^l |W_i||S_i| \leq (|W| - |W_1|)\sum_{i=2}^l |S_i| = (|W| - |W_1|)(|S| - |S_1|)$ deletions.  
    By considering the yet unmapped $|W_1|(|S| - |S_1|)$ insertions from the previous step, there is currently a total of at most $|W|(|S| - |S_1|)$ modifications to map.
    We map those to distinct modifications from group~\ref{grp:rs}.  
    
    Notice that there remains at least $|S| - |S_1|$ modifications in~\ref{grp:rs} that we can still map to.

    \item 
    \emph{(modifications incident to $T_1$)}: recall that each vertex of $T_1$ has at most one neighbor in $U_1$ and that no two vertices of $T_1$ have the same neighbor in $U_1$, by the definition of sister vertices.  Therefore in $\B'$, 
    there are at most $|U_1|$ deletions between vertices of $T_1$ and $U_1$.  We map those to the remaining $|U_1|$ modifications from group~\ref{grp:ruz} that we can still map to.  No other modification affects $T_1$ since $S \subseteq N_G(T_1)$.

    \item 
    \emph{(modifications incident to other $T_i$'s)}: 
    for $i \in \{2, 3, \ldots, l\}$, there are at most $|U_i|$ deletions between vertices of $T_i$ and $U_i$.  If $|U_i| = 0$, then no deletions need to be accounted for. 
    
    If $|U_i| \geq 2$, first observe that $|T_i| \geq |U_i|$.  Indeed, as $U_i = B_i \cap (N(T_i) \setminus S)$, it is a subset of the neighborhood of $T_i$, and therefore each element of $U_i$ has at least one neighbor in $T_i$. 
    Moreover, we have seen that each element of $T_i$ has at most one neighbor in $U_i$.  Thus, each vertex of $T_i$ contributes to adding at most one vertex in $U_i$, from which it follows that $|T_i| \geq |U_i|$.
    We can map the at most $|U_i|$ deletions between $T_i$ and $U_i$ to $|T_i|(|U_i|-1) \geq |T_i| \geq |U_i|$ distinct insertions from group~\ref{grp:tu}.

    Finally, assume that $|U_i| = 1$.  By the definition of a sister, only one vertex of $T_i$ has the element of $U_i$ in its neighborhood (otherwise $T_i$ would have twins). 
 Therefore, we have only a single deletion to map.  If $S_i \neq S$, map the deletion to one of the $|T_i|(|S| - |S_i|) \geq 1$ deletions from group~\ref{grp:ts} (here we may assume $|S| > 0$ as otherwise $R$ has no neighbors and $T$ would be empty).
    If instead $S_i = S$, then map the deletion to one of remaining $|S| - |S_1|$ modifications from group~\ref{grp:rs} (this is possible since $|S| - |S_1| = |S| \geq 1$, and this is used only once, for the $i$ when $S_i = S$, so the map remains injective).

\end{itemize}

Since all vertices of $V_1$ belong to one of the $R, T_i, W_i$, or $Y_i$ sets, we have mapped the edge modifications incident to all vertices of $V_1$.
We have thus provided an injective map from $E(G) \triangle E(\B')$ to $E(G) \triangle E(\B)$, implying that the former is not larger than the latter.  Therefore, $\B'$ is also optimal.  Moreover, one of its clusters is $R \cup T \cup S \cup W_1$, implying that all sisters of $R$ have undergone a deletion to their vertex outside of $S$.

To conclude the argument, note that $\B'$ can be obtained by first taking any $u \in T$, deleting its incident edge shared with the vertex $v$ not in $S$, and then applying at most $k - 1$ more modifications to obtain $\B'$.  Thus $G - uv$ can be transformed into $\B'$ with cost at most $k - 1$.
\qed
\end{proof}

We now have enough ingredients to devise our small kernel.  

\begin{theorem}\label{thm:5k}
Let $G$ be a graph on which Rules 1,2 and 3 do not apply, and suppose that there exists a biclustering of $G$ of cost at most $k$.  Then $G$ has at most $4.5k$ vertices.
\end{theorem}

\begin{proof}
    
Let us first make an observation. 
Let $R$ be a twin class of $G$, let $T$ be the sisters of $R$ and let $W := N(N(R)) \setminus (R \cup T)$.
Since Rules 1,2 and 3 do not apply, we know that $|R| \leq |W|$.
Indeed, suppose that $|R| > |W|$.
If $|T| \geq 1$, then Rule 3 applies, so assume $T = \emptyset$.
In that case, $W = N(N(R)) \setminus R$, but then, Rule 2 applies.  We will thus assume that $|R| \leq |W|$ for each twin class $R$.

Let $\B$ be an optimal biclustering of $G$ in which vertices in the same twin class of $G$ belong to the same bicluster, which exists by Lemma~\ref{lem:twins-in-same}.  Let $E_I := E(\B) \setminus E(G)$ be the set of inserted edges and let $E_D := E(G) \setminus E(\B)$ be the set of deleted edges. Note that by assumption, $|E_I \cup E_D| \leq k$.  Let $X$ be the set of vertices of $G$ that are incident to at least one edge of $E_I \cup E_D$.
Notice that for a twin class $R$ of $G$, either none or all vertices of $R$ are in $X$ (because if one member of $R$ is incident to a modified edge, we must apply the same modification to all members of $R$).
Let $\{R_1, \ldots, R_p\}$ be the set of twin classes of $G$ that are not contained in $X$.  Thus each $R_i$ is unmodified (i.e. not incident to a modified edge), and therefore each member of each $R_i$ has the same neighborhood in $\B$ as in $G$.   
For each $R_i$, let $S_i := N_G(R_i)$,  
let $T_i$ be the sisters of $R_i$ and $W_i := N_G(S_i) \setminus (R_i \cup T_i)$.
Moreover, let $B_i$ be the set of vertices in the bicluster of $\B$ that $R_i$ belongs to (note that $T_i$ and $W_i$ are not necessarily contained in $B_i$).
Since for each $i \in [p]$, $R_i$ is unmodified, we have $S_i \subseteq B_i$ and, since each $W_i \cup T_i$ vertex has neighbors in $S_i$ but a different neighborhood than the $R_i$ vertices, we have $W_i \cup T_i \subseteq X$.  

We first argue that no two distinct $R_i, R_j$ classes can be in the same bicluster.

\begin{nclaim}\label{claim:basickernel}
    Let  $a \in \{1,2\}$  and let $R_i, R_j \subseteq V_a$, where $i, j \in [p]$ and $i \neq j$.  
    Then $B_i \neq B_j$ and $S_i \cap S_j = \emptyset$.
\end{nclaim}

\begin{proof}
    If $B_i = B_j$ was true, because $R_i$ and $R_j$ are unmodified and are assumed to be in the same $V_a$, we would have $N_G(R_i) = N_G(R_j)$ and $R_i \cup R_j$ would be a larger twin class.  Thus $B_i \neq B_j$.  Moreover, $S_i \subseteq B_i$ and $S_j \subseteq B_j$ implies $S_i \cap S_j = \emptyset$.
    \qed
\end{proof}

Observe that although the $R_i$'s and $S_i$'s on the same side do not intersect, it is possible that the $T_i$'s and $W_i$'s on the same side do intersect.

We bound the  size of $V := V(G)$ by devising a charging scheme in which each vertex of $V$ charges a total amount of $1$ to edges of $E_I \cup E_D$.  The main idea is that every modified edge receives a total charge of at most $4.5$ in this scheme. 
To this end, first define a function $f : V(G)  \rightarrow X$ as follows:
\begin{itemize}
    \item 
    For each $x \in X$, put $f(x) = x$. 

    \item 
    For each $R_i$, $i \in [p]$, map each $r \in R_i$ to some $w \in W_i$ in an arbitrarily injective manner.  More precisely, define $f$ so that $f(r) \in W_i$ for each $r \in R_i$, and such that $f(r) \neq f(r')$ for each distinct $r, r' \in R_i$.  Note that this is possible since $|R_i| \leq |W_i|$.
\end{itemize}
Note that $f$ is not necessarily injective, but $f$ restricted to any $R_i$ is injective.  
Also note that because each $R_i$ is on the same side of the bipartition as $W_i$, $f(u)$ is always on the same side as $u$.
We will often use the notation $I_f(x) = \{u : f(u) = x\}$ for the inverse of $f$.
See Figure~\ref{fig:fmap}.

Now define the following charging scheme: for $x \in X$, let $M(x) \subseteq E_I \cup E_D$ be the set of modified edges incident to $x$.
Then each $u \in V(G)$ charges an amount of $1/|M(f(u))|$ to each edge of $M(f(u))$.  In other words, $u$ spreads a total charge of $1$ across all modified edges incident to $f(u)$. 
Note that for a modified edge $xy$, the total charge received by the edge is $|I_f(x)|/|M(x)| + |I_f(y)|/|M(y)|$.

\begin{figure}
    \centering
    \includegraphics[width=.3\textwidth]{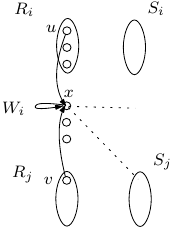}
    \caption{An illustration of the $f$ map.  Here, $x \in X$ and $I_f(x) = \{u, x, v\}$.  No vertex of $R_i$ other than $u$ maps to $x$, and no vertex of $R_j$ other than $v$ maps to $x$.  Dotted lines indicate modified edges incident to $x$.  Each of $u, x,$ and $v$ charges $1/2$ to each dotted edge.}
    \label{fig:fmap}
\end{figure}

\begin{nclaim}\label{claim:kerneldels}
    Let $x \in X$.  Then the following holds: 
    \begin{enumerate}
        \item 
        the elements of $I_f(x) \setminus \{x\}$ are all in different biclusters of $\B$.  

        \item 
        for each bicluster $B$ of $\B$ that does not contain $x$ but that intersects with $I_f(x) \setminus \{x\}$, there is a deleted edge incident to $x$ whose other endpoint is in $B$.

        \item 
        there are at least $|I_f(x)| 
 - 2$ deleted edges that are incident to $x$.
    \end{enumerate}
\end{nclaim}

\begin{proof}
    Let $v_1, v_2, \ldots, v_h$ be the vertices in $I_f(x) \setminus \{x\}$.  Notice that none of these vertices are in $X$, as otherwise they would map to themselves.  Hence they are all in some unmodified twin classes.  
    For distinct $i, j \in [h]$, let $R_{a_i}$ and $R_{a_j}$ be the twin classes that contain $v_i$ and $v_j$, respectively. 
    The fact that $f(v_i) = f(v_j) = x$ implies that $v_i$ and $v_j$ are on the same side of the bipartition of $G$, and so are $R_{a_i}$ and $R_{a_j}$. 
    Morever, $R_{a_i} \neq R_{a_j}$, since $f$ restricted to $R_{a_i}$ or $R_{a_j}$ is injective.  Then Claim~\ref{claim:basickernel} implies that $B_{a_i} \neq B_{a_j}$, which proves our first point. 
    
    For the second part, suppose that some bicluster does not contain $x$ but contains some element $u$ such that $f(u) = x$.  Since $u$ is from some unmodified twin class, we must have $u \in R_i$ for some $i \in [p]$.  Hence $u \in B_i$ and $x \notin B_i$.  By the construction of $f$, we have $x \in W_i$, which means that $x$ has a neighbor in $S_{i}$. Since $S_{i}$ is contained in $B_{i}$ and $x \notin B_{i}$, there must be at least one deletion between $x$ and $S_{i}$.  

    The third part is a consequence of the two previous.  By the first point, at most two elements of $I_f(x)$ are in the same bicluster as $x$, namely $x$ and some other vertex.  The other vertices are all in different biclusters, and by the second point each require at least one distinct deletion incident to $x$.
    \qed
\end{proof}

\begin{nclaim}
    Let $xy \in E_I \cup E_D$ such that at least one of $x$ or $y$ is incident to an inserted edge.  Then $xy$ receives a total charge of at most $4.5$.
\end{nclaim}

\begin{proof}
Suppose without loss of generality that $x \in V_1, y \in V_2$, and that $x$ is incident to an inserted edge.
By Claim~\ref{claim:kerneldels}, $x$ is incident to at least $|I_f(x)| - 2$ deleted edges and, counting the insertion, to $|M(x)| \geq |I_f(x)| - 1$ modified edges.  

Suppose first that $xy$ is itself an inserted edge.
If $|I_f(x)| = 1$, then only $x$ sends charge to $xy$ on the $V_1$ side, and so the $V_1$ side sends a total charge of at most $1$ to $xy$.  If $|I_f(x)| \geq 2$, the total charge sent to $xy$ by the $V_1$ side is 
 \[
    \dfrac{|I_f(x)|}{|M(x)|} \leq \dfrac{|I_f(x)|}{|I_f(x)| - 1} \leq 2.
\]
Since $y$ is also incident to an inserted edge, by the same argument, the total charge sent to $y$ by the $V_2$ side is at most $2$, and thus $xy$ receives a charge of $4$ or less.

So suppose that $xy$ is a deleted edge.
Then $x$ is incident to at least two modified edges (and also to at least $|I_f(x)| - 1$ modified edges as before).  If $|I_f(x)| \leq 2$, this means that $V_1$ sends a charge of at most $1$ to $xy$.  Otherwise if $|I_f(x)| \geq 3$, this total charge is at most $|I_f(x)|/(|I_f(x)| - 1) \leq 3/2$.
As for the $V_2$ side, if $|I_f(y)| \leq 2$, then it charges at most $2$ to $xy$ and $3/2 + 2 < 4.5$.  Otherwise, with $|I_f(y)| \geq 3$, 
using Claim~\ref{claim:kerneldels} (third part) to upper bound the charge of the $V_2$ side, the total charge sent to $xy$ is at most 
\[
\frac{3}{2} + \dfrac{|I_f(y)|}{|I_f(y)| - 2} \leq \frac{3}{2} + 3 = 4.5.
\]
    \qed
\end{proof}

We now restrict our attention to vertices of $X$ that are only incident to edges of $E_D$.
For the remainder of the proof, let $xy \in E_D$ such that none of $x$ or $y$ is incident to an inserted edge, with $x \in V_1$ and $y \in V_2$.  
Observe that if $|I_f(x)| \leq 2$, then $xy$ receives a charge of at most $2$ from the $V_1$ side.  Otherwise if $|I_f(x)| \geq 3$, then $V_1$ sends a charge of at most $|I_f(x)|/(|I_f(x)| - 2)$.  This quantity is $3$ if $|I_f(x)| = 3$, and at most $2$ if $|I_f(x)| \geq 4$.  The same observations hold for $y$ and the $V_2$ side.  It follows that if none of $|I_f(x)|$ or $|I_f(y)|$ is $3$, the total charge is at most $4$.
Let us assume, without loss of generality, that $|I_f(x)| = 3$.  
If $x$ is incident to at least two deletions, then the total charge sent by the vertices of $I_f(x)$ is at most $1/2 \cdot 3 = 3/2$.  Given that the charge from the $V_2$ side is at most $3$, the total is at most $4.5$.  

We can thus assume that $x$ is incident to only $xy$ as a deleted edge and that $|I_f(x)| = 3$.  Let $\{x, u, v\} = I_f(x)$.  By Claim~\ref{claim:kerneldels}, $u$ and $v$ are in distinct biclusters and by the same claim, one of $u$ or $v$, say $u$, is in the same bicluster  as $x$ (otherwise we would need two deletions incident to $x$).  Let $R_i$ be the twin class containing $u$.  Now, $x$ is in the same bicluster as $S_i$ and is incident to no inserted edge, which means that $S_i \subseteq N_G(x)$.  Moreover, $N_G(x) \setminus S_i$ has only one element, namely $y$, as otherwise more deletions would be required.  We also have $x \in W_i$, which means that $x$ is not a sister of $u$.  This is only possible if $x$ has a twin $x'$.  Since $\B$ keeps twins together, $y$ must be incident to the two deleted edges $xy$ and $x'y$.

Let us now bound the charge sent to $xy$ by the $V_2$ side.  If $|I_f(y)| \leq 2$, the total charge sent to $xy$ by the $V_2$ side is at most $1$ since $y$ is incident to at least two deletions.  With the charge of $3$ from the $V_1$ side, we get a charge of at most $4$ for $xy$.  So we will assume $|I_f(y)| \geq 3$.  Let $v_1, \ldots, v_l$ be the members of $I_f(y) \setminus \{y\}$, and let $B_{a_1}, \ldots, B_{a_l}$ be the biclusters that respectively contain them.  By Claim~\ref{claim:kerneldels}, 
$y$ is incident to at least $|I_f(y)| - 2$ deletions, each with an endpoint in a different bicluster among $B_{a_1}, \ldots, B_{a_l}$.  Moreover, since $x$ and $x'$ are in the same bicluster $B_i$, we know that two deletions are between $y$ and $B_i$, which means that $y$ is incident to at least $|I_f(y)| - 1$ deletions.
Assuming $|I_f(y)| \geq 3$, the quantity $|I_f(y)|/(|I_f(y)| - 1)$ is at most $3/2$.  Adding this to the charge of $3$ from the $V_1$ side, the total charge is again $4.5$.

We have thus argued that any of the $k$ edges in $E_I \cup E_D$ receives a charge of at most $4.5$.
Since each vertex of $G$ sends a charge of $1$ and the total charge is at most $4.5k$, it follows that the number of vertices is at most $4.5k$. 
\qed 
\end{proof}

One can also show that the above analysis is tight.  Figure~\ref{fig:tight} is an example for $k = 2$ with $9 = 4.5k$ vertices.  It is the graph obtained by starting with a $P_6$, then adding a twin for every vertex on the $V_1$ side.

\begin{figure}
    \centering
    \includegraphics[width=0.15\textwidth]{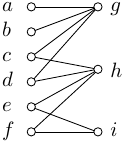}
    \caption{A graph of size $4.5k$, with $k = 2$, on which none of the reduction rules apply.}
    \label{fig:tight}
\end{figure}

\begin{theorem}
    For any even $k \geq 2$, there exist graphs on which Rules 1,2, and 3 are not applicable, whose optimal biclustering requires $k$ edge modifications and that contain $4.5k$ vertices.
\end{theorem}

\begin{proof}
    Consider $k = 2$ and the graph in Figure~\ref{fig:tight}.  The vertices on the left side can be partitioned into twin classes $\{a, b\}, \{c, d\}$, and $\{e,f\}$, and those on the right side into $\{g\}, \{h\}, \{i\}$.  Since we may assume that twins undergo the same modifications, we see that at least $k = 2$ modifications are needed, since at least one twin class from the left side needs to be incident to an edge modification.
    We also see that a biclustering can be obtained with the $2$ deletions $ch$ and $dh$.
    Rule 1 obviously does not apply.  Rule 2 does not either, since no twin class $R$ has size greater than $N(N(R))$.  Finally, Rule 3 does not apply because, on the left side, any vertex has a twin and thus cannot be a sister, and on the right side, any two vertices have at least two neighbors in their symmetric difference.
    Finally, since the graph has $9$ vertices and $k = 2$, this is a kernel with $4.5k$ vertices.  For larger $k$, we can take the disjoint union of $k/2$ copies of this graph.
    \qed
\end{proof}

\section{An $O^*(2.581^k)$ branching algorithm}

Before presenting our main branching algorithm, we introduce some notation that is specific to this section.
A \emph{branching vector} is a sequence $B = (b_1, \ldots, b_t)$ of positive integers (not necessarily sorted, and with possible repetition).  The \emph{branching factor} of $B$ is the largest real root of the polynomial $f(x) = x^{k} - \sum_{i=1}^t x^{k - b_i}$, where $k$ is the parameter of interest.  We call $f(x)$ the \emph{characteristic polynomial} of $B$.  A \emph{branching algorithm} is a recursive algorithm in which every recursive call runs in polynomial time, and that either solves its given instance, or branches into recursive  calls that each reduce the value of the parameter $k$ by at least $1$.  The reduction in $k$ of each branch results in a branching vector, and it is well known that if $a$ is the supremum of the branching factors of all the branching vectors encountered during the execution, then the algorithm runs in time $O^*(a^k)$, where again $O^*$ hides polynomial factors (see e.g.~\cite[Chapter 2]{Fomin2010}).
We say that a branching vector $B = (b_1, \ldots, b_t)$ is \emph{better} than a branching vector $B' = (b'_1, \ldots, b'_s)$ if there is an injective map $\sigma : [t] \rightarrow [s]$ such that $b_i \leq b'_{\sigma(i)}$ for each $i \in [t]$.  For example, $(1, 2, 3, 3, 4)$ is better than $(1, 2, 2, 3, 3, 4)$.    Note that if $B$ is better than $B'$, then the branching factor of $B$ is smaller than or equal to that of $B'$.  On the other hand, if the branching factor of $B$ is smaller than that of $B'$, it does not necessarily imply that $B$ is better than $B'$.

A central concept in our algorithm are conflicting vertices.

\begin{definition}
    Let $G = (V_1 \cup V_2, E)$ be a bipartite graph and let $u, v \in V_i$ for some $i \in \{1, 2\}$.  We say that $u$ and $v$ are \emph{in conflict} if $N(u) \cap N(v) \neq \emptyset$ and $N(u) \triangle N(v) \neq \emptyset$.
\end{definition}

Note that a graph $G$ is a bicluster graph if and only if if does not contain two vertices that are in conflict.  We also have the following.

\begin{lemma}\label{lem:connected-component-bicluster}
Let $G = (V_1 \cup V_2, E)$ be a bipartite graph and let $u \in V_i$.  Suppose that $u$ is not in conflict with any vertex of $G$.  Then the connected component that contains $u$ is a bicluster.
\end{lemma}

\begin{proof}
Suppose $u \in V_1$, without loss of generality.   Let $R_u$ be the twin class that contains $u$.  If $N(u)$ or $N(N(u))$ is empty, the lemma easily holds.
Otherwise, let $v \in N(N(u))$.  Then $N(v) \cap N(u) \neq \emptyset$.  Since $u$ has no conflict, we have $N(u) \triangle N(v) = \emptyset$.  Therefore, $N(u) = N(v)$ and $v \in R_u$.  
In other words, $N(N(R_u)) = R_u$, from which the statement follows.
\qed
\end{proof}

The idea of our algorithm is to first remove connected components that are biclusters (by Rule 1).  Then we consider any two vertices $u$ and $v$ that are in conflict.  Note that by Lemma~\ref{lem:connected-component-bicluster}, we can choose any $u \in V(G)$ and we will find some $v$ it is in conflict with.
We then observe that in any biclustering of $G$, either $u$ and $v$ belong to the same bicluster, or they do not.  We branch on these two possibilities.
If we fix the choice that $u, v$ are in the same bicluster, we need to do something about  $N(u) \triangle N(v)$, and we can try every way of ensuring that $u$ and $v$ have the same neighbors.  
If we fix the choice that $u, v$ are in different biclusters, we need to do something about $N(u) \cap N(v)$, and this time we ensure that $u$ and $v$ share no common neighbor.   See Figure~\ref{fig:alg1}.

For technical reasons, in Algorithm~\ref{alg:mainbranching} we bound the number of elements of $N(u) \triangle N(v)$ and $N(u) \cap N(v)$ that we branch on by a constant $100$.
This is because if, say, $|N(u) \triangle N(v)|$ was a function of $k$, branching over every way of enforcing $u$ and $v$ to have the same neighbors could take time that is superpolynomial with respect to $k$, and we would perform the same large number of recursive calls.  
Although it can be argued that this is not problematic because the decrease in $k$ compensates sufficiently, restricting to a constant does not alter the worst-case scenario.  More importantly, it allows us to state that each recursive call takes polynomial time, which simplifies the analysis.
In the algorithm, the choice of $C$ and $D$ is arbitrary.

\begin{algorithm}[h]
%\SetAlgoNoLine%
\DontPrintSemicolon
\SetKwProg{Fn}{function}{}{}
\Fn{biclusterize($G, k$)}{
	\lIf{$k < 0$}
	{
	    Report ``NO'' and return
	}
	Remove from $G$ all bicluster connected components (Rule 1)\;
	\lIf{$G$ has no vertex}
	{
	    Report ``YES'' and return
	}
	\lIf{$G$ has maximum degree $2$}
	{
	    Solve $G$ in polynomial time
	}
	\;

	Let $u, v \in V(G)$ such that $|N(u)| \geq 3$, and such that $N(u) \cap N(v) \neq \emptyset$ and  $N(u) \triangle N(v) \neq \emptyset$ \; %, chosen so that $N(u) \cup N(v)$ is maximum\;
	Let $R_u$ be the twin class of $u$ and $R_v$ be the twin class of $v$\;
	%\;
	/*Put $u, v$ in the same bicluster*/\;
        Let $D \subseteq N(u) \triangle N(v)$ of size $\max(100, |N(u) \triangle N(v)|)$\;
	\For{each subset $Z$ of $D$} 
	{
	    Obtain $G'$ from $G$ by: \hspace{200mm} 
	    inserting all missing edges
	    between $R_u \cup R_v$ and $Z$ and \hspace{200mm} 
	    deleting all  edges between $R_u \cup R_v$ and $D \setminus Z$\;
	    Let $h$ be the number of edges modified from $G$ to $G'$\;
	    $biclusterize(G', k - h)$\;
	}
	/*Put $u, v$ in different biclusters*/\;
        Let $C \subseteq N(u) \cap N(v)$ of size $\max(100, |N(u) \cap N(v)|)$\;
	\For{each subset $Z$ of $C$}
	{
	    Obtain $G'$ from $G$ by: \hspace{200mm} 
	    deleting all edges between $R_u$ and $C \setminus Z$ and \hspace{200mm} 
	    deleting all edges between $R_v$ and $Z$\;
	    Let $h$ be the number of edges modified from $G$ to $G'$\;
	    $biclusterize(G', k - h)$\;
	}
	\lIf{some recursive call reported ``YES''}{report ``YES''} \lElse{report ``NO''}
  }
  \caption{Main branching algorithm}
  \label{alg:mainbranching}
\end{algorithm}

The pseudo-code of this approach is presented in Algorithm~\ref{alg:mainbranching}.
As we will show, it is sufficient to achieve time $O^*(3.237^k)$ with an easy analysis.  
%We then show that, with a more involved analysis and one single special case handling, this can be reduced to $O^*(2.695^k)$.

\begin{figure}[t]
    \centering
    \includegraphics[width=.65\textwidth]{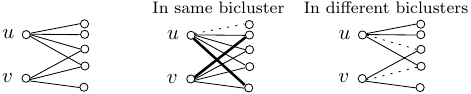}
    \caption{Left: two vertices $u, v$ and $N(u) \cup N(v)$.  Middle: one of the $8$ ways to branch into if $u$ and $v$ are in the same bicluster.  Right: one of the $4$ ways to branch into if $u$ and $v$ are in distinct biclusters.  Dotted edges are deletions, fat edges are insertions.}
    \label{fig:alg1}
\end{figure}

Importantly, every time the algorithm applies $h$ modifications according to some choice of $Z$, we observe that $h \geq |D|$ in the first loop of branching because each $z \in Z$ requires an insertion and each $z \notin Z$ requires a deletion (non-equality is possible when one of the twin classes has more than one member).  Similarly, $h \geq |C|$ in the second loop, because each element requires a deletion.  

Notice that even if Algorithm~\ref{alg:mainbranching} might branch into an exponential number of cases (with respect to the sizes of $C$ and $D$), the larger the set of branching cases, the more $k$ is reduced in each recursive call.  It turns out to be more advantageous to have more cases: the larger $N(u) \cap N(v)$ and $N(u) \triangle N(v)$ are, the closer to an $O(2^k)$ algorithm we get.  The graphs of small maximum degree are the most problematic.   

To analyze this more formally, let $u, v \in V_i$ and let $c := |C|$ and $d := |D|$, where  $C \subseteq N(u) \cap N(v)$ and $D \subseteq N(u) \triangle N(v)$ are the sets used by the algorithm for the branching ($c$ stands from `common', $d$ for `different').  Then in the worst case, when $|R_u| = |R_v| = 1$, the number of recursive calls is given by the recurrence
$$
f(k) = 2^c f(k - c) + 2^d f(k - d),
$$
whose characteristic polynomial is 
$a^k - 2^c a^{k - c} - 2^d a^{k - d}$.
We will be interested in the worst possible combination of $c$ and $d$, i.e. that lead to a polynomial with the maximum largest real root.

\begin{definition}
Let $f$ be a polynomial function.  Then $lrr(f)$ denotes the largest real root of $f$.

Let $c, d \geq 1$ be integers.
We denote by $lrr(c, d)$ the largest real root of the polynomial function $f(a) = a^k - 2^c a^{k - c} - 2^d a^{k - d}$.
\end{definition}

The intuition that higher $c$ and $d$ is better follows from the next technical lemma, which describes some situations in which a characteristic polynomial is preferable over another.  
Although we will not need the full generality of the statement here, it might be of independent interest.

\begin{lemma}\label{lem:high-c-is-better}
Let $k \in \mathbb{N}$, let $b > 1$ be a real number and let $c, c_0, \ldots, c_{k-1}$ be non-negative reals, with $0 < c < k$.  Let 

\[ 
f(a) = a^k - \sum_{j = 0}^{k - 1} c_j a^j - b^c a^{k - c}
\]

be a polynomial.
Moreover let $\epsilon$ be such that $0 < \epsilon \leq k - c$, and let 
\[
f^*(a) = a^k - \sum_{j = 0}^{k - 1} c_j a^j - b^{c + \epsilon} a^{k - c - \epsilon}
\]
Then $lrr(f^*) \leq lrr(f)$.
\end{lemma}

\begin{proof}
We first claim that $f^*(a) \geq f(a)$ for all $a \geq b$, where $b$ is as defined in the statement.
We prove this by contraposition, i.e. we assume that ${f^*(a) < f(a)}$ for some $a$, and deduce $a < b$.  If $a \leq 0$, then $a < b$, so suppose $a$ is positive.
Then by assumption,

\vspace{-5mm}
\begin{align*}
    a^k - \sum_{j = 0}^{k - 1} c_j a^j - b^{c + \epsilon} a^{k - c - \epsilon} < a^k - \sum_{j = 0}^{k - 1} c_j a^j - b^c a^{k - c},
\end{align*}
and thus $- b^{c + \epsilon} a^{k - c - \epsilon} < - b^c a^{k - c}$.
%
%\begin{align*}
    %- b^{c + \epsilon} a^{k - c - \epsilon} < - b^c a^{k - c}
%\end{align*}
%
This leads to $a^{\epsilon} < b^{\epsilon}$.  Given that $a > 0, b > 1$, this means that $a < b$, proving our claim by contraposition.

We next claim that $lrr(f) \geq b$.
Note that for each $a$ satisfying $0 \leq a < b$, because the $c_j$'s are non-negative,
$f(a) \leq a^k - b^c a^{k-c} = a^{k - c}(a^c - b^c)$.
Still assuming that $0 \leq a < b$, we have $a^c - b^c < 0$ and thus the real roots of $f$ cannot be in the interval $[0..b)$.  Moreover, $f$ is negative in this interval and, since the leading coefficient of $f$ is positive, $f$ is eventually positive for some $a \geq b$, and by the continuity of $f$, $f(a) = 0$ for some $a \geq b$.  This implies $lrr(f) \geq b$.

Owing to its leading coefficient, we also know that $f$ is positive for all ${a > lrr(f)}$.  
Since $f^*(a) \geq f(a)$ for all $a \geq b$ and $lrr(f) \geq b$, 
we know that $f^*$ is also positive for all $a > lrr(f)$.
Thus $f^*(a) > 0$ for all $a > lrr(f)$. It follows that any real root of $f^*$ must be equal to or less than $lrr(f)$, proving the lemma.
\qed
\end{proof}

\begin{corollary}\label{cor:lrr-better-high-c-d}
Let $c, d \geq 1$ be integers.  Then $lrr(c, d) \geq lrr(c + 1, d)$ and $lrr(c, d) \geq lrr(c, d + 1)$.
\end{corollary}

\begin{proof}
Recall that $lrr(c, d)$ is the largest real root of 
${f(a) = a^k - 2^c a^{k-c} - 2^d a^{k - d}}$.  Consider $lrr(c, d+1)$,  the largest real root of ${f^*(a) = a^k - 2^c a^{k-c} - 2^{d + 1} a^{k-d-1}}$.  Using Lemma~\ref{lem:high-c-is-better} and plugging in $b = 2$, $\epsilon = 1$, we see that $lrr(c, d + 1) \leq lrr(c, d)$.  The proof is the same for $lrr(c+1, d)\leq lrr(c, d)$. 
\qed
\end{proof}

We can then show that the simple algorithm above achieves a similar running time as that of~\cite{guo2008improved}.
The time complexity proof is based on the idea that, the worst case occurs when $c$ and $d$ are as small as possible.  Since $G$ contains a vertex of degree at least $3$, the worst case is when we branch over $u, v$ of degrees $3$ and $1$, which gives branching factor $lrr(1, 2) \simeq 3.237$.  The details follow.

\begin{theorem}\label{thm:mainalgo}
Algorithm~\ref{alg:mainbranching} is correct and runs in time $O^*(3.237^k)$.
\end{theorem}

\begin{proof}
Let us first prove the correctness of the algorithm.  We use induction on $k$.  As a base case, if $k < 0$, then reporting ``NO'' is correct since there can exist no solution that modifies a negative number of edges.
Assume $k \geq 0$, and suppose the algorithm is correct for smaller values.
If $G$ has maximum degree $2$, then $G$ consists of cycles and paths and can easily be solved optimally.
If $G$ contains a connected component $X$ that is a bicluster, then Rule 1 lets us remove $X$ safely. 
If $G$ is has no vertices after this, then no modification is needed and reporting ``YES'' is correct.
If $G$ still has vertices, then each of its connected component is not a bicluster.  By contraposition of Lemma~\ref{lem:connected-component-bicluster}, for any vertex $u \in V_i$ with $i \in \{1,2\}$, there exists $v \in V_i$ such that $u$ and $v$ are in conflict.
By choosing $u$ to be of degree at least $3$, this guarantees that the algorithm will find $u, v$ as described.

Now, assume that it is possible to edit at most $k$ edges of $G$ to obtain a bicluster graph.
Let $R_u$ and $R_v$ be the twin classes of $u$ and $v$, respectively.  Let $\B$ be an optimal biclustering in which all twins are in the same bicluster (invoking Lemma~\ref{lem:twins-in-same} once again).  Note that the cost of $\B$ is at most $k$, and that we may thus assume that all members of $R_u$, and all members of $R_v$, are in the same bicluster.  

Suppose first that $R_u$ and $R_v$ are in the same bicluster $B$.  Let $D$ be the subset of $N(u) \triangle N(v)$ chosen by the algorithm.
Consider $Z = B \cap D$.
Then to obtain $\B$, we must insert every missing edge between $R_u \cup R_v$ and $Z$, and delete every edge between $R_u \cup R_v$ and $D \setminus B = D \setminus Z$.  Let $G'$ be the graph obtained from $G$ after applying these modifications, and let $h$ be the number of modifications.  We know that $G'$ can be modified into $\B$ using at most $k - h$ modifications.
Since $Z \subseteq D$, in the first loop of branchings, we branch on the $Z$ subset, in particular.  
Therefore, the algorithm recurses on $G'$ and $k - h$ at some point. By induction, the algorithm will correctly report ``YES'', and it will also do so on input $G$ and $k$.

Assume instead that $R_u, R_v$ are not in the same bicluster of $\B$.
Let $B_u$ and $B_v$ be the biclusters that contain $R_u$ and $R_v$, respectively.  Let $C$ be the subset of $N(u) \cap N(v)$ chosen by the algorithm.
Consider $Z = B_u \cap C$.
To obtain $\B$, we must delete every edge between $R_u$ and $C \setminus Z$, and every edge between $R_v$ and $Z$.  Let $G'$ be the graph obtained from $G$ after applying these deletions, and let $h$ be the number of deletions.  We know that $G'$ can be modified into $\B$ using at most $k - h$ modifications.
Since $Z \subseteq C$, in the second loop of branchings we branch on the $Z$ subset, in particular.  
Thus the algorithm recurses on $G'$ and $k - h$ at some point. By induction, the algorithm will correctly report ``YES'', and it will also do so on input $G$ and $k$.

Finally, assume instead that it is not possible to edit at most $k$ edges of $G$ to obtain a bicluster graph.
Algorithm~\ref{alg:mainbranching} branches into several ways of obtaining an alternate graph $G'$, each time reducing $k$ by the number $h$ of modifications to turn $G$ into $G'$.
We know that for every such $G'$, it is not possible to apply at most $k - h$ modifications to $G'$ to obtain a bicluster graph, as otherwise $G$ would admit a solution with at most $k$ modifications.  Therefore, every recursive call will report ``NO'', and thus the algorithm will correctly report ``NO'' on $G$.  This proves correctness.

As for the time complexity, 
we may assume that $G$ has maximum degree at least $3$ and that its connected components are not all biclusters, as otherwise $G$ is solvable in polynomial time.  
%Also assume that Rule 1, 2 and 3 have been applied on $G$ (which clearly takes polynomial time).

Let $u$ be a vertex of degree at least $3$, and assume $u \in V_1$, without loss of generality.  Let $v \in V_1$ be such that $N(u) \cap N(v) \neq \emptyset \neq N(u) \triangle N(v)$.
As argued previously, such a vertex $v$ must exist by contraposition of Lemma~\ref{lem:connected-component-bicluster}.
Let $c = |C|$ and $d = |D|$, where $C \subseteq N(u) \cap N(v)$ and $D \subseteq N(u) \triangle N(v)$ are the sets chosen by the algorithm.  By the definition of a conflict, neither set is empty.
Consider the first branching loop, which branches over $2^d$ possibilities.  For each $Z \subseteq D$ considered, we must modify at least $d$ edges since each member of $Z$ requires an insertion and each member of $D \setminus Z$ requires a deletion.
We therefore have $h \geq d$ in every iteration, and we branch over $2^d$ possibilities that each reduce $k$ by at least $d$.

Consider the second branching loop.
For each $Z \subseteq C$ considered, we must delete at least $c$ edges since each $z \in Z$ requires the $vz$ deletion and each $z' \in C \setminus Z$ requires the $uz'$ deletion.
We therefore have $h \geq c$ in every iteration, and in the worst case we branch over $2^c$ possibilities that each reduce $k$ by $c$.

This yields a branching vector with $2^c$ subcases with value $k - c$ or less, plus $2^d$ subcases with value $k - d$ or less, corresponding to a worst-case branching factor of $lrr(c, d)$.
Since $u$ is of degree at least $3$, 
we have $c + d \geq 3$.
By Corollary~\ref{cor:lrr-better-high-c-d}, we may restrict our attention to the smallest possible values of $c$ and $d$ that satisfy this inequality, 
which is either $c = 1, d = 2$ or $c = 2, d = 1$.  
The complexity follows from $lrr(1, 2) = lrr(2, 1)$, and that the largest real root of $a^k - 2a^{k-1} - 4a^{k - 2}$ is bounded by $3.237$.  
Since each individual call in the recursions takes polynomial time with respect to $k$ and the size of $G$ (in particular because we bound the size of $C$ and $D$ by a constant), this results in an $O^*(3.237^k)$ time algorithm.
\qed
\end{proof}

\subsection*{With a little bit more work: an $O^*(2.581^k)$ time algorithm}

One might notice that the worst case complexity of the algorithm is achieved when the maximum degree of $G$ is low, since high degree vertices will lead to better branching factors.
This complexity can be significantly improved by handling low degree vertices in an ad-hoc manner.  
In fact, only a slight modification of the algorithm is necessary.  We only need to handle degree $1$ vertices in a particular manner, 
and to restrict ourselves to optimal biclusterings that favor deletions over insertions.  
These allow us to attain a time complexity of $O^*(2.581^k)$.  Although this only requires simple modifications to the algorithm, it is the analysis that is somewhat more complex, as we need to dive into branching calls that create degree 1 vertices or twins.

To be more specific, consider any vertex $v$ and its degree $deg(v)$ in the original input graph.  Then one option is to have $v$ in its own bicluster and delete every edge incident to $v$.  Therefore it is pointless to have more than $deg(v)$ modified edges with $v$ as an endpoint.  
Furthermore, if exactly $deg(v)$ modified edges contain $v$ and some of them are insertions, we can always replace these by $deg(v)$ deletions without altering the optimality of the solution.  %This is formalized next.

\begin{lemma}\label{lem:solution-with-most-dels}
For a bipartite graph $G$, there exists an optimal biclustering of $G$ such that for any $v \in V(G)$, either $deg(v)$ modified edges contain $v$ and they are all deletions, or at most $deg(v) - 1$ modified edges contain $v$.
\end{lemma}

\begin{proof}
Let $\B$ be an optimal biclustering of $G$ in which the number of deleted edges is maximum, among all optimal biclusterings.  Suppose for contradiction  that for some $v \in V(G)$, it is not the case that $deg(v)$ deleted edges contain $v$, but at least $deg(v)$ modified edges contain $v$.  Using Observation~\ref{obs:sumcost}, we may assume that exactly $deg(v)$ modified edges contain $v$, as otherwise we could put $v$ in its own bicluster to get a better biclustering.
Consider obtaining a biclustering $\B'$ from $\B$ by removing $v$ from its bicluster, and adding a new bicluster that only contains $v$.  Then in $\B'$, $deg(v)$ modified edges contain $v$, so $\B'$ is also optimal.  Moreover, $\B'$ must have strictly more deletions than $\B$, since not all modified edges containing $v$ in $\B$ were deletions.  This contradicts our choice of $\B$.
\qed
\end{proof}

The modified algorithm remembers the original graph and, after any modification, checks that we meet the requirements of Lemma~\ref{lem:solution-with-most-dels}.  It aborts if this is not the case, which is correct by the above lemma. 
The next step is to handle degree $1$ vertices.  %The actual algorithm can be inferred from the proof, but in the Appendix, we provide the full details of the algorithmic procedure.

\begin{lemma}\label{lem:degree-one}
Let $G$ be a bipartite graph on which Rule 1 does not apply, and suppose that $G$ has a vertex of degree $1$.  Then it is possible to achieve a branching vector of $(1, 2, 3, 3, 4)$ or better, and thus branching factor less than $2.066$.
\end{lemma}

\begin{proof}
Assume that $G$ has a vertex $u \in V_i$ with a single neighbor $v$, $i \in \{1,2\}$.  
By Lemma~\ref{lem:solution-with-most-dels}, we may assume that either $uv$ gets deleted, or that no modified edge contains $u$.
Let $R_u$ be the twin class that contains $u$.
Let $W = N(v) \setminus R_u$, and observe that each $w \in W$ has degree at least $2$ (otherwise, $w$ would be in $R_u$).  If $W$ is empty, then $R_u \cup \{v\}$ forms a bicluster and Rule 1 applies, so suppose $W \neq \emptyset$.
We consider four cases in the proof, which lead to either an immediate deletion, or branching vectors $(1,2), (1,2,3)$, or $(1,2,3,3,4)$, which are all better than $(1, 2, 3, 3, 4)$. 

\vspace{1mm}

\noindent
\emph{Case 1: $W$ has a single vertex $w$, and $w$ has degree $2$.}
In this case, notice that $w$ is a sister of $u$, and that Rule 3 applies since all neighbors of $v$ are sisters of $u$.  We may thus delete the edge between $w$ and its neighbor other than $v$ and reduce $k$ by $1$.

\vspace{1mm}

\noindent
\emph{Case 2: $|W| \geq 2$ and $W$ has a vertex of degree $2$.}
Let $w \in W$ be of degree $2$ and let $z$ be the neighbor of $w$ other than $v$.
We can first branch into the situation where we delete $uv$ and decrease $k$ by $1$.  Otherwise, we consider that $uv$ is fixed and cannot be deleted.  We argue that $u$ and $w$ can be assumed to be in the same bicluster.
Suppose that in an optimal biclustering $\B$, $u$ and $w$ are in different biclusters. Notice that since no inserted edge contains $u$, the bicluster $B$ that contains $u$ contains only one vertex from the side of $v$, which is $v$ itself.
It follows that we can obtain an alternate biclustering $\B'$ by adding $w$ into $B$.  This requires deleting $wz$ instead of $wv$.  Thus, assuming that $uv$ is fixed, there exists an optimal biclustering in which $u$ and $w$ are in the same bicluster.  We can thus delete $wz$ safely, which reduces $k$ by $1$ and makes $u$ and $w$ twins.
Now let $w' \in W \setminus \{w\}$.  If $w'$ also has degree $2$, we can apply the same argument on the updated graph and delete the edge between $w'$ and its neighbour other than $v$.  In this situation, we branched over two cases, one with the deletion $uv$, and one with two deletions incident to $W$, with branching vector $(1, 2)$.

Suppose instead that $w'$ has degree $3$ or more.  We can put $u, w'$ in the same bicluster, in which case we must delete all edges between $w'$ and its neighbors other than $v$ (because we assume that $u$ is incident to no insertion).  Counting the $wz$ deletion, this makes at least $3$ deletions.  Or, we put $u, w'$ in distinct biclusters, in which case we delete $w'v$ since $uv$ is fixed, and in this scenario we make at total of two deletions. 
To sum up, we branch into $3$ cases with branching vector $(1,2,3)$.

\vspace{1mm}

\noindent
\emph{Case 3: $W$ has a single vertex $w$, and $w$ has degree $3$ or more.}
If $u$ and $w$ are in the same bicluster, we must delete the edges between $w$ and $N(w) \setminus \{v\}$ since we forbid insertions that involve $u$.  Since $w$ has degree at least $3$ or more, this decreases $k$ by at least $2$. 

If $u$ and $w$ are not in the same bicluster, we can either delete $uv$ or $wv$.  We claim that we can branch on only deleting $wv$. Let $\B$ be an optimal biclustering in which twins are in the same bicluster (which exists by Lemma~\ref{lem:twins-in-same}).  Assume that in $\B$, $u$ and $v$ are not in the same bicluster, so that $uv$ is deleted (as well as all the edges between $R_u$ and $v$).
We may assume that the bicluster $B_u$ of $\B$ that contains $u$ is equal to $R_u$.  Consider the biclustering $\B'$ obtained by removing $v$ from its bicluster and adding $v$ to $B_u$.  This saves any inserted edge containing $v$, and saves all the deletions between $R_u$ and $v$, which is at least one deletion since $uv$ is such an edge.
On the other hand, this requires deleting $wv$ (but no other edge since $|W| = 1$).  We save at least as many modifications as we create, so we may assume that $wv$ is deleted.  Thus we only branch in the case where $wv$ is deleted, reducing $k$ by $1$.
The above cases lead to branching vector $(1,2)$.

\vspace{1mm}

\noindent
\emph{Case 4: $|W| \geq 2$ and only contains vertices of degree $3$ or more.}
Let $x,y \in W$.  
We branch into five scenarios that cover all ways of obtaining a solution\footnote{Notice that this could be divided into three cases: two of $x, y$ are with $u$, only one of them is, or none.  However, our division into five cases makes the recursion calls to make more explicit.}: 

\vspace{-3mm}

\begin{itemize}
    \item 
    delete $uv$ and reduce $k$ by $1$;
    
    \item 
    fix $uv$, and put $u, x$ in the same bicluster.
    This requires deleting at least two edges incident to $x$ (since it has degree $3$ or more, and we forbid insertions with $u$).
    Once this is done, we consider two more cases:
        \begin{itemize}
            \item 
            put $u, y$ in the same bicluster.
            This also requires deleting at least two edges incident to $y$.  When we reach that case, we can decrease $k$ by at least $4$.
            
            \item 
            put $u, y$ in different biclusters.
            This requires deleting $yv$, and we may decrease $k$ by at least $3$.
        \end{itemize}
        
    \item 
    fix $uv$, and put $u, x$ in  different biclusters.
    This requires deleting $xv$.
    Once this is done, we consider two more cases:
        \begin{itemize}
            \item 
            put $u, y$ in the same bicluster.
            This requires deleting at least two edges incident to $y$.  When we reach that case, we can decrease $k$ by at least $3$.
            
            \item 
            put $u, y$ in different biclusters.
            This requires deleting $yv$, and we may decrease $k$ by $2$.
        \end{itemize}
\end{itemize}

These five cases yield branching vector $(1, 2, 3, 3, 4)$, which has branching factor less than $2.066$.
\qed
\end{proof}

Recall that the definition of ``better'' branching vectors implies that even though the degree 1 rule may use several branching vectors, they are all, in some sense, subsequences of $(1,2,3,3,4)$.  Therefore, it will be safe to assume that the degree 1 rule always achieves this branching vector. 

The rest of the time analysis is dedicated to show that if we get rid of vertices of degree one, we are able to achieve a much better branching factor. 
Note that Lemma~\ref{lem:degree-one} requires a graph on which Rule 1 does not apply.
Thus we must ensure that the vertices of degree $1$ that we create do not get deleted by Rule 1.  This will be handled by the following lemma.

\begin{lemma}\label{lem:deg-one-no-bicluster}
Let $G$ be a bipartite graph.  Suppose that $G$ has minimum degree at least $2$.  Let $u \in V(G)$ and let $v \in N(u)$.  Let $G'$ be obtained from $G$ by deleting every edge incident to $u$ except $uv$.  
Then the connected component of $G'$ that contains $u$ is not a bicluster.
\end{lemma}

\begin{proof}
Note that $N_G(v) = N_{G'}(v)$ and that $v$ has some neighbor $w$ other than $u$.  Furthermore, $N_G(w) = N_{G'}(w)$ and $w$ has some neighbor $z$ other than $v$, and $z \notin N_{G'}(u)$.  Thus $u, w$ are in conflict and the connected component containing $u$ cannot be a bicluster by Lemma~\ref{lem:connected-component-bicluster}.
\qed
\end{proof}

We next observe that all the branching cases that put $u, v$ in the same bicluster create twins.  When such twins have remaining conflicts, we can branch on them to improve the analysis.

\begin{lemma}\label{lem:twins}
    Let $G$ be a bipartite graph on which Rule 1 does not apply.  Suppose that $G$ has a twin class with at least two vertices.  Then we may achieve branching vector $(1,2,2,3,3,4)$ or better, and therefore branching factor less than $2.317$.
\end{lemma}

\begin{proof}
    If $G$ has a vertex of degree 1, then we may apply the branching vector $(1,2,3,3,4)$ from Lemma~\ref{lem:degree-one}, which is better than $(1,2,2,3,3,4)$.  
    Hence, we may assume that $G$ has minimum degree $2$.
    Let $R$ be a twin class with at least two members, and let $u$ be a vertex such that $N(R) \cap N(u)$ and $N(R) \triangle N(u)$ are non-empty.
    
    Suppose that there exists $z \in N(R) \setminus N(u)$.
    Next, assume that $N(R) \cap N(u)$ has at least two vertices $v, w$.   We can branch on putting $R, u$ in the same bicluster, one case deleting at least $2$ edges between $R$ and $z$, the other case inserting $1$ edge $uz$.
    Then we branch on putting $R, v$ in distinct biclusters, having $4$ branches for the possible deletions between $R, u$ and $v, w$.  Since all members of $R$ undergo the same deletions, the branches respectively make $2, 3, 3, 4$ deletions.  This yields branching vector $(1, 2, 2, 3, 3, 4)$.
    Assume instead that $N(R) \cap N(u)$ has only one vertex $v$.  Since $deg(u) \geq 2$,  $N(u) \setminus N(R)$ has at least one vertex $w$.  Branching on separating $R$ and $u$ makes $1$ and $2$ deletions because of $v$, and branching on putting $R, u$ in the same bicluster makes $2,3,3,4$ modifications because of the four ways to handle $v, w$, with the same branching vector $(1,2, 2,3,3,4)$.  

    Finally, suppose that $z$ does not exist, i.e. $N(R) \subseteq N(u)$.  Then since $|N(R)| \geq 2$, there exist $v, w \in N(R) \cap N(u)$, and there exists $z \in N(u) \setminus N(R)$.  This yields the same branching vector as the first case from the previous paragraph.
    \qed
\end{proof}

Together, our analysis of degree $1$ vertices and twin classes allow us to handle each corner case of Algorithm~\ref{alg:mainbranching} to improve their complexity.

\begin{lemma}\label{lem:d-geq-2}
    Let $G$ be a bipartite graph on which Rule 1 does not apply.  Suppose that $G$ has no vertices of degree $1$ and no twins.
    Suppose further that there exist $u, v \in V(G)$ such that $N(u) \cap N(v) \neq \emptyset$ and $|N(u) \triangle N(v)| \geq 2$.  Then we may achieve branching factor less than $2.581$.
\end{lemma}

\begin{proof}
Let $c := |N(u) \cap N(v)|$ and $d := |N(u) \triangle N(v)| \geq 2$.  We consider all possible values of $c$ on a case-by-case basis.

\vspace{1mm}

\noindent
\emph{Case 1: $c \geq 4$.}  By Corollary~\ref{cor:lrr-better-high-c-d}, Algorithm~\ref{alg:mainbranching} achieves its worst case branching when $d = 2, c = 4$, which is $lrr(4, 2) \leq 2.55$.

\vspace{1mm}

\noindent
\emph{Case 2: $c = 3$.}
If $d \geq 3$, then Algorithm~\ref{alg:mainbranching} can achieve branching factor $lrr(3, 3) \leq 2.52$.
So assume $c = 3$ and $d = 2$.  Let $x, y$ be distinct vertices of $N(u) \cap N(v)$.  Notice that since $x, y$ are not twins but have common neighbors, they are in conflict.  Now, branching on putting $u, v$ in distinct biclusters creates eight branches with $3$ deletions each.  
Branching on putting $u, v$ in the same bicluster requires four branches with $2$ modifications each, all incident to vertices in $N(u) \triangle N(v)$.  In each such branch, $x$ and $y$ are still in conflict because their neighborhood has not changed and, moreover, $u$ and $v$ have become twins.  It follows that the connected component containing $u$ and $v$ in these branches is not a bicluster and will not be removed by Rule 1.  Furthermore, in those four branches, Lemma~\ref{lem:twins} is applicable, in which case we may apply branching vector $B = (1, 2, 2, 3, 3, 4)$ or better.  
%By our definition of ``better'', if a different branching vector $B'$ is achieved in some of those recursive calls, then every entry of $B'$ is present in $B$ due to the injective map.  We may thus assume that branching vector $B$ is achieved (for instance by adding dummy recursive calls to achieve branching vector exactly $B$).  
Since we already made two modifications before applying this rule, the actual branching vector of these four branches is $(3, 4, 4, 5, 5, 6)$ or better.  
This amounts to a branching vector in which $3$ appears $12$ times, $4$ and $5$ appear $8$ times, and $6$ appears $4$ times, or a better branching vector.  Since better branching vectors result in branching factors that are not worse, the former is the worst case, and it has branching factor less than $2.56$.

\vspace{1mm}

\noindent 
\emph{Case 3: $c = 2$}. 
If $d \geq 4$, Algorithm~\ref{alg:mainbranching} achieves branching factor $lrr(2, 4) \leq 2.55$.
Suppose that $d = 3$ and, without loss of generality, that $|N(u) \setminus N(v)| \geq |N(v) \setminus N(u)|$.  There are four ways with $2$ deletions to put $u, v$ in distinct bicluster, and eight ways with $3$ modifications to put them in the same bicluster.  
Suppose $N(u) \setminus N(v)$ has $3$ vertices.  Since $c = 2$ and $d = 3$, we infer that $deg(v) = 2$.  Of the eight ways to put $u, v$ in the same bicluster, there are four that require at least $2$ insertions incident to $v$ which we may assume do not occur by Lemma~\ref{lem:solution-with-most-dels}.  This yields branching vector $(2,2,2,2,3,3,3,3)$ and branching factor less than $2.39$.   
So suppose that $N(u) \setminus N(v)$ has 2 vertices, and $N(v) \setminus N(u)$ has the other.  Thus $v$ has degree $3$.  We consider only seven ways to put $u, v$ in the same bicluster (with $3$ modifications each), since one way requires two insertions and one deletion incident to $v$, and we may discard it by Lemma~\ref{lem:solution-with-most-dels}.   
Moreover, one of the four way to put $u, v$ in distinct biclusters  deletes two edges incident to $v$, which makes $v$ of degree $1$.  By Lemma~\ref{lem:deg-one-no-bicluster}, we may apply the degree $1$ rule after the $2$ deletions, leading to branching vector $(3, 4, 5, 5, 6)$ or better in that case.  
We thus have seven branches with $3$ modifications; three branches with $2$ deletions; one branch that yields vector $(3, 4, 5, 5, 6)$ or better.  This gives branching factor less than $2.54$.

To finish the case, suppose that $d = 2$.  The four ways with $2$ modifications that put $u, v$ in the same bicluster each produce twins and, as in the previous case, we may assume that they are present in a connected component that is not a bicluster (since $c = 2$ and the two common neighbors remain in conflict, as before).  Thus all four of these branches lead to branching vector $(3, 4, 4, 5, 5, 6)$ or better.  
If $|N(u) \setminus N(v)| = 2$, one of these four ways can be omitted since $deg(v) = 2$ and we cannot make $v$ incident to two insertions.  In that case, the branching vector has four occurrences of $2$ (for $u, v$ in distinct biclusters) and three times the vector $(3, 4, 4, 5, 5, 6)$ with branching factor less than $2.56$.  
So assume that $|N(u) \setminus N(v)| = |N(v) \setminus N(u)| = 1$.
Of the four ways to put $u, v$ in distinct biclusters, two create a degree $1$ vertex by deleting only edges incident to $u$, or to $v$.  By Lemma~\ref{lem:deg-one-no-bicluster}, these yield branching vector $(3, 4, 5, 5, 6)$.  
Thus, we have four occurrences of $(3,4,4,5,5,6)$ due to twins, two occurrences of $2$, and two occurrences of $(3,4,5,5,6)$, with branching factor $2.581$.

\vspace{1mm}

\noindent 
\emph{Case 4: $c = 1$}.  If $d \geq 6$, then we know that the worst case occurs when $d = 6$, and $lrr(1, 6) \leq 2.58$.  Thus assume that $d \leq 5$.
Also assume without loss of generality that $|N(u) \setminus N(v)| \geq |N(v) \setminus N(u)|$.  
There are two branches with $1$ deletion each to put $u, v$ in distinct biclusters.
If $d = 2$, then $deg(u) = deg(v) = 2$. Two of the four ways of putting $u, v$ in the same bicluster, each having $2$ modifications, require an insertion plus a deletion both incident to $u$, or both incident to $v$, which we assume do not occur by Lemma~\ref{lem:solution-with-most-dels}.  Moreover, the two ways of separating $u, v$ create a degree $1$ vertex after $1$ deletion and, using Lemma~\ref{lem:degree-one} and Lemma~\ref{lem:deg-one-no-bicluster}, yield branching vector $(2, 3, 4, 4, 5)$ or better.  This yields branching vector $(2, 2, 2, 3, 4, 4, 5, 2, 3, 4, 4, 5)$ and branching factor $2.39$.  

Assume $d \geq 3$.  Suppose that $|N(v) \setminus N(u)| = 1$, so that $deg(v) = 2$.  When putting $u, v$ in the same bicluster, we keep at most one member of $N(u) \setminus N(v)$ in that bicluster, otherwise $v$ requires two insertions (we again invoke Lemma~\ref{lem:solution-with-most-dels}).  By the same lemma, if we keep one such member, we must also keep the neighbor of $v$ in $N(v) \setminus N(u)$.  Finally, if we keep no member of $N(u) \setminus N(v)$, we cannot keep the neighbor of $v$, as otherwise $u$ would be incident to $deg(u)$ modified edges with one insertion.
Thus we only need to branch on $|N(u) \setminus N(v)| \leq d$ cases (keep one exclusive neighbor of $u$, or none) each requiring $d$ modifications.  For $d = 3,4,5$, this gives respective branching vectors $(1,1,3,3,3), (1,1,4,4,4,4), (1,1,5,5,5,5,5)$ all with branching factor at most $2.49$.  

Suppose that $|N(v) \setminus N(u)| = 2$.  If $|N(u) \setminus N(v)| = 2$, 
$deg(u) = deg(v) = 3$.  Let $\{a_1, a_2\} = N(u) \setminus N(v)$ and $\{b_1, b_2\} = N(v) \setminus N(u)$.  There are 8 ways to put $u, v$ in the same bicluster that satisfy Lemma~\ref{lem:solution-with-most-dels} (if we keep both $a_1, a_2$, we must keep both $b_1, b_2$ because of $v$, if we keep only $a_1$, we must keep $b_1$ or $b_2$, or both, if we keep $a_2$ the same holds, and if we do not keep $a_1$ nor $a_2$, we cannot keep $b_1$ nor $b_2$ because of $u$).  Each of these four branch has $4$ modifications and, counting the two ways to separate $u$ and $v$, this leads to branching factor $2.51$.  If $|N(u) \setminus N(v)| = 3$, let $\{a_1, a_2, a_3\} = N(u) \setminus N(v)$.  There are at most $16$ ways to put $u, v$ in the same bicluster while satisfying Lemma~\ref{lem:solution-with-most-dels}.  To see this, observe that we cannot keep all of $a_1, a_2, a_3$; if we keep two of them we must keep both $b_1, b_2$ (three ways), if we keep one of $a_1, a_2, a_3$ we must keep one of $b_1$ or $b_2$ ($3 \cdot 3 = 9$ ways), and if we keep none there are four options for $b_1$ and $b_2$.  Each branch has $5$ modifications and this leads to branching factor $2.45$.

The case $|N(u) \setminus N(v)| \geq 4$ cannot occur at this point because we assume $d \leq 5$.  Similarly,  the case $|N(v) \setminus N(u)| \geq 3$ cannot occur, since $|N(u) \setminus N(v)| \geq |N(v) \setminus N(v)|$, and we assume $d \leq 5$.
\qed
\end{proof}

\begin{lemma}\label{lem:d-equal-1}
    Suppose that $G$ has minimum degree at least $2$ and maximum degree at least $3$, and suppose that $G$ has no twins.  Suppose that for every pair of vertices $u, v$ that are in conflict, we have $|N(u) \triangle N(v)| = 1$.  

    Then $G$ has maximum degree exactly $3$.  Moreover, 
     we may achieve branching vector $(1, 1)$ and thus branching factor $2$.
\end{lemma}

\begin{proof}
    
    Let $u$ be a vertex of maximum degree and let $v$ be a vertex in conflict with $u$.  
    Suppose that $deg(u) \geq 4$.  Since $|N(u) \triangle N(v)| = 1$, we have $N(v) \subset N(u)$ and $deg(v) = deg(u) - 1 \geq 3$.  
    Let $x, y, z$ be three distinct common neighbors of $u$ and $v$, and suppose without loss of generality that $deg(x) \geq deg(y) \geq deg(z)$. 
    Since none of $x, y, z$ are twins and since they have neighbors in common, each pair must be in conflict.  
    This implies that $|N(x) \triangle N(y)| = 1$ and thus that $deg(x) = deg(y) + 1$.  By the same reasoning, $deg(y) = deg(z) + 1$.  Thus $deg(x) = deg(z) + 2$, a contradiction since $|N(x) \triangle N(z)| = 1$ is not possible.

    Thus we may assume that $deg(u) = 3$.  Let $x, y, z$ be its neighbors, and assume $N(v) = \{x, y\}$.  Then $x, y$ must be in conflict, so we may assume without loss of generality that $x$ has a neighbor $z'$ that $y$ does not have (and no other neighbor because the maximum degree is $3$).  
    We claim that there exists an optimal biclustering in which $u, v$ are in the same bicluster.  
    Consider an optimal biclustering $\B$ in which this is not the case. Then there are at least two deleted edges that are incident to $u$ or $v$ because of $x$ and $y$.  
    Consider the biclustering $\B'$ obtained from $\B$ by removing $u, v, x, y$ from their respective biclusters, and adding the bicluster $\{u, v, x, y\}$.  The only modifications that could be in $\B'$ but not in $\B$ are the deletions between $uz$ and $xz'$.  On the other hand, $\B$ has at least two modifications not in $\B'$, which are the aforementioned deletions.  Hence $\B'$ is also optimal (another way to see this is as follows: when we temporarily remove $\{u, v, x, y\}$ from the biclustering, we save at least two deletions, and when we reincorporate it, we introduce exactly two deletions).  We may therefore assume that $u, v$ are in the same bicluster, which only requires two branches (delete $uz$ or insert $vz$).
    \qed
\end{proof}

\begin{theorem}
    The \textsc{Bicluster Editing} problem can be solved in time $O^*(2.581^k)$.
\end{theorem}

\begin{proof}
    We assume that $G$ has no connected component that is a bicluster.  If $G$ has maximum degree $2$ or less, then Algorithm~\ref{alg:mainbranching} solves $G$ in polynomial time.  
    If $G$ has a vertex of degree $1$, we can achieve branching factor at most $2.066$ by Lemma~\ref{lem:degree-one}.
    Otherwise, if $G$ has twins, we may achieve branching factor at most $2.317$ by Lemma~\ref{lem:twins}.
    If $G$ has no vertex of degree $1$ and no twins, but has conflicting vertices $u, v$ with $|N(u) \triangle N(v)| \geq 2$, then Lemma~\ref{lem:d-geq-2} applies and we get branching factor $2.581$.  If no such pair of vertices exists, every conflicting $u$ and $v$ satisfy $|N(u) \triangle N(v)| = 1$ and Lemma~\ref{lem:d-equal-1} gives branching factor $2$.  The worst branching factor is therefore $2.581$.
    \qed
\end{proof}

\noindent
\textbf{Concluding remarks.}  We conclude with some open ideas.  First, one could check whether our branching algorithm approach could be improved by combining it with the half-edge ideas used in \textsc{Cluster Editing} (see~\cite{bocker2012golden}).  Second, we do not know the true complexity of Algorithm~\ref{alg:mainbranching}.  Its running time should be better than $O^*(2.581^k)$, since recursions at deeper levels will have twins classes with more and more vertices, and large twin classes have good branching factors.  However, these analyzes become increasingly technical and less interesting --- unless an elegant formalization can be devised.  
Let us also remind the reader that all our results use the assumption that the given graph $G$ is bipartite, and that no edge can be added between vertices on the same side.  To our knowledge, no linear kernel is known for the more general variant in which these assumptions are removed.
It is also open whether the $4.5k$ kernel can be improved with additional or different reduction rules.

\vspace{3mm}

\noindent
\textbf{Acknowledgements.}  ML acknowledges Benoît Charbonneau and Pierre-Luc Parent for their help in finding the $4.5k$ lower bound on the kernel size.

ML acknowledges financial support from the Natural Sciences and Engineering Research Council of Canada - Discovery grant.

\vspace{3mm}

\noindent
\textbf{Declarations.}  The authors declare that they have no competing interests.

\bibliographystyle{splncs04}

\bibliography{main}

\begin{thebibliography}{10}
\providecommand{\url}[1]{\texttt{#1}}
\providecommand{\urlprefix}{URL }
\providecommand{\doi}[1]{https://doi.org/#1}

\bibitem{ailon2012improved}
Ailon, N., Avigdor-Elgrabli, N., Liberty, E., Van~Zuylen, A.: Improved
  approximation algorithms for bipartite correlation clustering. SIAM Journal
  on Computing  \textbf{41}(5),  1110--1121 (2012)

\bibitem{altenhoff2019oma}
Altenhoff, A.M., Levy, J., Zarowiecki, M., Tomiczek, B., Vesztrocy, A.W.,
  Dalquen, D.A., M{\"u}ller, S., Telford, M.J., Glover, N.M., Dylus, D.,
  et~al.: Oma standalone: orthology inference among public and custom genomes
  and transcriptomes. Genome research  \textbf{29}(7),  1152--1163 (2019)

\bibitem{amit2004bicluster}
Amit, N.: The bicluster graph editing problem. Ph.D. thesis, Tel Aviv
  University (2004)

\bibitem{bansal2004correlation}
Bansal, N., Blum, A., Chawla, S.: Correlation clustering. Machine learning
  \textbf{56}(1-3),  89--113 (2004)

\bibitem{barber2007modularity}
Barber, M.J.: Modularity and community detection in bipartite networks.
  Physical Review E  \textbf{76}(6),  066102 (2007)

\bibitem{bocker2012golden}
B{\"o}cker, S.: A golden ratio parameterized algorithm for cluster editing.
  Journal of Discrete Algorithms  \textbf{16},  79--89 (2012)

\bibitem{cao2012cluster}
Cao, Y., Chen, J.: Cluster editing: Kernelization based on edge cuts.
  Algorithmica  \textbf{64},  152--169 (2012)

\bibitem{chen20122k}
Chen, J., Meng, J.: A 2k kernel for the cluster editing problem. Journal of
  Computer and System Sciences  \textbf{78}(1),  211--220 (2012)

\bibitem{cheng2000biclustering}
Cheng, Y., Church, G.M.: Biclustering of expression data. In: Ismb. vol.~8, pp.
  93--103 (2000)

\bibitem{dondi2023tractability}
Dondi, R., Lafond, M.: On the tractability of covering a graph with 2-clubs.
  Algorithmica  \textbf{85}(4),  992--1028 (2023)

\bibitem{drange_et_al:LIPIcs.IPEC.2015.402}
Drange, P.G., Reidl, F., S\'{a}nchez~Villaamil, F., Sikdar, S.: {Fast
  Biclustering by Dual Parameterization}. In: 10th International Symposium on
  Parameterized and Exact Computation (IPEC 2015). vol.~43, pp. 402--413 (2015)

\bibitem{Fomin2010}
Fomin, F.V., Kratsch, D.: Exact Exponential Algorithms. Springer Berlin
  Heidelberg (2010)

\bibitem{gnatyshak2012gaining}
Gnatyshak, D., Ignatov, D.I., Semenov, A., Poelmans, J.: Gaining insight in
  social networks with biclustering and triclustering. In: Perspectives in
  Business Informatics Research: 11th International Conference, BIR 2012,
  Nizhny Novgorod, Russia, September 24-26, 2012. Proceedings 11. pp. 162--171.
  Springer (2012)

\bibitem{gramm2003graph}
Gramm, J., Guo, J., H{\"u}ffner, F., Niedermeier, R.: Graph-modeled data
  clustering: Fixed-parameter algorithms for clique generation. In: Italian
  Conference on Algorithms and Complexity. pp. 108--119. Springer (2003)

\bibitem{gramm2004automated}
Gramm, J., Guo, J., H{\"u}ffner, F., Niedermeier, R.: Automated generation of
  search tree algorithms for hard graph modification problems. Algorithmica
  \textbf{39}(4),  321--347 (2004)

\bibitem{guo2008improved}
Guo, J., H{\"u}ffner, F., Komusiewicz, C., Zhang, Y.: Improved algorithms for
  bicluster editing. In: International Conference on Theory and Applications of
  Models of Computation. pp. 445--456. Springer (2008)

\bibitem{hartung2015structural}
Hartung, S., Komusiewicz, C., Nichterlein, A., Such{\`y}, O.: On structural
  parameterizations for the 2-club problem. Discrete Applied Mathematics
  \textbf{185},  79--92 (2015)

\bibitem{karrer2011stochastic}
Karrer, B., Newman, M.E.: Stochastic blockmodels and community structure in
  networks. Physical review E  \textbf{83}(1),  016107 (2011)

\bibitem{komusiewicz2012cluster}
Komusiewicz, C., Uhlmann, J.: Cluster editing with locally bounded
  modifications. Discrete Applied Mathematics  \textbf{160}(15),  2259--2270
  (2012)

\bibitem{lafond2020even}
Lafond, M.: Even better fixed-parameter algorithms for bicluster editing. In:
  Computing and Combinatorics: 26th International Conference, COCOON 2020,
  Atlanta, GA, USA, August 29--31, 2020, Proceedings. pp. 578--590. Springer
  (2020)

\bibitem{lafond2018accurate}
Lafond, M., Meghdari~Miardan, M., Sankoff, D.: Accurate prediction of orthologs
  in the presence of divergence after duplication. Bioinformatics
  \textbf{34}(13),  i366--i375 (2018)

\bibitem{liu2012editing}
Liu, H., Zhang, P., Zhu, D.: On editing graphs into 2-club clusters. In:
  Frontiers in Algorithmics and Algorithmic Aspects in Information and
  Management, pp. 235--246. Springer (2012)

\bibitem{madeira2004biclustering}
Madeira, S.C., Oliveira, A.L.: Biclustering algorithms for biological data
  analysis: a survey. IEEE/ACM transactions on computational biology and
  bioinformatics  \textbf{1}(1),  24--45 (2004)

\bibitem{newman2004finding}
Newman, M.E., Girvan, M.: Finding and evaluating community structure in
  networks. Physical review E  \textbf{69}(2),  026113 (2004)

\bibitem{pinheiro2016solving}
Pinheiro, R.G., Martins, I.C., Protti, F., Ochi, L.S., Simonetti, L.G.,
  Subramanian, A.: On solving manufacturing cell formation via bicluster
  editing. European Journal of Operational Research  \textbf{254}(3),  769--779
  (2016)

\bibitem{pontes2015biclustering}
Pontes, B., Gir{\'a}ldez, R., Aguilar-Ruiz, J.S.: Biclustering on expression
  data: A review. Journal of biomedical informatics  \textbf{57},  163--180
  (2015)

\bibitem{protti2006applying}
Protti, F., da~Silva, M.D., Szwarcfiter, J.L.: Applying modular decomposition
  to parameterized bicluster editing. In: International Workshop on
  Parameterized and Exact Computation. pp. 1--12. Springer (2006)

\bibitem{rosvall2008maps}
Rosvall, M., Bergstrom, C.T.: Maps of random walks on complex networks reveal
  community structure. Proceedings of the National Academy of Sciences
  \textbf{105}(4),  1118--1123 (2008)

\bibitem{de2012hybrid}
de~Sousa~Filho, G.F., Lucidio~dos Anjos, F.C., Ochi, L.S., Protti, F.: Hybrid
  metaheuristic for bicluster editing problem. Electronic Notes in Discrete
  Mathematics  \textbf{39},  35--42 (2012)

\bibitem{de2017new}
de~Sousa~Filho, G.F., J{\'u}nior, T.L.B., Cabral, L.A., Ochi, L.S., Protti, F.:
  New heuristics for the bicluster editing problem. Annals of Operations
  Research  \textbf{258}(2),  781--814 (2017)

\bibitem{sun2013biclue}
Sun, P., Guo, J., Baumbach, J.: Biclue-exact and heuristic algorithms for
  weighted bi-cluster editing of biomedical data. In: BMC proceedings. vol.~7,
  p.~S9. Springer (2013)

\bibitem{sun2014complexity}
Sun, P., Guo, J., Baumbach, J.: Complexity of dense bicluster editing problems.
  In: International Computing and Combinatorics Conference. pp. 154--165.
  Springer (2014)

\bibitem{tanay2005biclustering}
Tanay, A., Sharan, R., Shamir, R.: Biclustering algorithms: A survey. Handbook
  of computational molecular biology  \textbf{9}(1-20),  122--124 (2005)

\bibitem{tsur2021faster}
Tsur, D.: Faster parameterized algorithm for bicluster editing. Information
  Processing Letters  \textbf{168},  106095 (2021)

\bibitem{tsur2023faster}
Tsur, D.: Faster parameterized algorithms for bicluster editing and flip
  consensus tree. Theoretical Computer Science  \textbf{953},  113796 (2023)

\bibitem{van2009deterministic}
Van~Zuylen, A., Williamson, D.P.: Deterministic pivoting algorithms for
  constrained ranking and clustering problems. Mathematics of Operations
  Research  \textbf{34}(3),  594--620 (2009)

\end{thebibliography}

\end{document}